\DeclareSymbolFont{AMSb}{U}{msb}{m}{n}
\documentclass[reqno,centertags,11pt,a4paper,noamsfonts]{amsart}
\pdfoutput=1
\usepackage[svgnames]{xcolor}
\usepackage{graphicx}
\usepackage[english]{babel}
\usepackage[babel=true,expansion=alltext,protrusion=alltext-nott,final]{microtype}
\usepackage[margin={3cm},marginparwidth={2.5cm}]{geometry}
\usepackage[utf8]{inputenc}
\usepackage{textcomp}
\usepackage[sb]{libertine}
\usepackage[libertine,bigdelims,vvarbb]{newtxmath}
\useosf
\usepackage[varqu,varl]{inconsolata}
\usepackage{mathtools}
\usepackage[T1]{fontenc}
\usepackage{textcomp}
\usepackage{amsthm}
\usepackage[inline]{enumitem}
\numberwithin{equation}{section}
\usepackage[normalem]{ulem}
\usepackage{tikz}
\usetikzlibrary{arrows}
\usepackage{pgfplots}
\pgfplotsset{compat=1.15}
\usepackage{bm}
\usepackage{sidecap}

\usepackage{amsmath}
\DeclareFontFamily{U}{mathx}{}
\DeclareFontShape{U}{mathx}{m}{n}{<-> mathx10}{}
\DeclareSymbolFont{mathx}{U}{mathx}{m}{n}
\DeclareMathAccent{\widehat}{0}{mathx}{"70}
\DeclareMathAccent{\widecheck}{0}{mathx}{"71}

\providecommand{\mr}[1]{\href{http://www.ams.org/mathscinet-getitem?mr=#1}{MR~#1}}
\providecommand{\zbl}[1]{\href{https://zbmath.org/?q=an:#1}{Zbl~#1}}

\providecommand{\doi}[2]{\href{https://doi.org/#1}{#2}}
\newcommand{\RR}{\mathbb{R}}
\newcommand{\CC}{\mathbb{C}}

\newcommand{\C}{\mathcal{C}}
\newcommand{\1}{\mathbf{1}}
\newcommand{\ii}{\mathrm{i}}

\renewcommand{\Re}{\mathrm{Re}\,}

\DeclareMathOperator{\weaklystar}{\rightharpoonup\kern-2.2ex ^* \, \,}
\DeclareMathOperator{\sgn}{sign}

\def\XXint#1#2#3{{\setbox0=\hbox{$#1{#2#3}{\int}$ }
\vcenter{\hbox{$#2#3$ }}\kern-.6\wd0}}

\newcommand{\R}{\mathbb R}
\newcommand{\N}{\mathbb N}
\newcommand{\Z}{\mathbb Z}
\renewcommand{\C}{\mathbb C}

\newcommand\norm[1]{\lVert #1 \rVert}
\newcommand\bignorm[1]{\big\lVert #1 \big\rVert}

\newcommand\inner[1]{\langle #1 \rangle}

\newcommand{\im}{{\rm{Im}}}

\newcommand{\ra}{\rightarrow}

\newcommand{\re}{\mathrm{Re}\,}
\renewcommand{\phi}{\varphi}
\newcommand{\dom}{\mathrm{dom}\,}
\newcommand{\signum}{\mathrm{sign}\,}

\theoremstyle{plain}
\newtheorem{theorem}{Theorem}[section]
\newtheorem{proposition}[theorem]{Proposition}
\newtheorem{corollary}[theorem]{Corollary}
\newtheorem{lemma}[theorem]{Lemma}

\newtheorem*{theorem*}{Theorem}

\theoremstyle{definition}
\newtheorem{definition}[theorem]{Definition}
\newtheorem{remark}[theorem]{Remark}
\newtheorem*{remark*}{Remark}

\usepackage[pagebackref=false]{hyperref}
\hypersetup{
	linktoc=all,
	bookmarksdepth=3,
	colorlinks=true,
	linkcolor=Green,citecolor=Orange,urlcolor=DarkBlue,
	pdftitle={On a variational problem related to the CLR inequality},
	pdfauthor={Thiago Carvalho Corso, Tobias Ried}, 
	pdfsubject={},
	pdfkeywords={}} 
\begin{document}
\numberwithin{table}{section}
\title[On a variational problem related to the CLR and LT inequalities]{On a variational problem related to the Cwikel--Lieb--Rozenblum and Lieb--Thirring inequalities}
\author[T.~Carvalho~Corso]{Thiago Carvalho Corso}
\address[T.~Carvalho Corso]{Institute of Applied Analysis and Numerical Simulation, University of Stuttgart, Pfaffenwaldring 57, 70569 Stuttgart, Germany}
\email{thiago.carvalho-corso@mathematik.uni-stuttgart.de}

\author[T.~Ried]{Tobias Ried}
\address[T.~Ried]{Max-Planck-Institut für Mathematik in den Naturwissenschaften, Inselstraße 22, 04103 Leipzig, Germany}
\curraddr{School of Mathematics, Georgia Institute of Technology, 686 Cherry Street, Skiles Building, Atlanta, GA 30332-0160 United States of America}
\email{tobias.ried@gatech.edu}
\keywords{Schrödinger operators, eigenvalues, bound states, Lieb--Thirring inequality, Cwikel--Lieb--Rozenblum inequality, Hadamard three lines lemma, Hardy spaces, maximal Fourier multipliers, boundary values of holomorphic functions}
\subjclass[2020]{Primary 35P15; Secondary 81Q10, 30D05}
\date{\today}
\thanks{\emph{Funding information}: DFG -- Project-ID 195170736 -- TRR109; DFG -- Project-ID 442047500 -- SFB 1481.\\[1ex]
\textcopyright 2024 by the authors. Faithful reproduction of this article, in its entirety, by any means is permitted for noncommercial purposes.}
\begin{abstract}
We explicitly solve a variational problem related to upper bounds on the optimal constants in the Cwikel--Lieb--Rozenblum (CLR) and Lieb--Thirring (LT) inequalities, which has recently been derived in \cite{HKRV23} and \cite{FHJN21}. We achieve this through a variational characterization of the $L^1$ norm of the Fourier transform of a function and duality, from which we obtain a reformulation in terms of a variant of the Hadamard three lines lemma. By studying Hardy-like spaces of holomorphic functions in a strip in the complex plane, we are able to provide an analytic formula for the minimizers, and use it to get the best possible upper bounds for the optimal constants in the CLR and LT inequalities achievable by the method of \cite{HKRV23} and \cite{FHJN21}. 
\end{abstract}

\maketitle
\tableofcontents
\setcounter{secnumdepth}{3}
\section{Introduction}
Motivated by the recent new proof of the Cwikel--Lieb--Rozenblum (CLR) inequality \cite{HKRV23}, we study the variational problem
\begin{align}
    M_\gamma \coloneqq \inf_{m_1, m_2 \in L^2(\R_+,\tfrac{\mathrm{d}s}{s})} \biggr\{\biggr(\norm{m_1}_{L^2(\R_+;\tfrac{\mathrm{d}s}{s})} \norm{m_2}_{L^2(\R_+;\tfrac{\mathrm{d}s}{s})}\biggr)^{\gamma-2} \int_0^\infty |m_1\star m_2(t) -t|^2 \, t^{-\gamma} \frac{\mathrm{d}t}{t}\biggr\}, \label{eq:primal}
\end{align}
for $\gamma >2$, where $\R_+ = (0,\infty)$ is the multiplicative group,
\begin{align*}
    &L^2(\R_+;\tfrac{\mathrm{d}s}{s}) = \biggr \{ f: \R_+ \rightarrow \C \mbox{ measurable }: \int_0^\infty |f(s)|^2 \frac{\mathrm{d}s}{s} < \infty \biggr\},
\end{align*}
and the convolution is taken with respect to the Haar measure $\tfrac{\mathrm{d}s}{s}$ in the multiplicative group,
\begin{align}\label{eq:convolution}
    (m_1 \star m_2)(t) = \int_0^\infty m_1\left(\frac{t}{s}\right) m_2(s) \frac{\mathrm{d}s}{s}.
\end{align}
The solution of \eqref{eq:primal} gives an upper bound on the constant in the CLR inequality.\footnote{For more details on the CLR estimate and how it is connected to the variational problem \eqref{eq:primal}, we refer to \cite{HKRV21,HKRV23}.} In the article \cite{HKRV23}, the authors proved a lower bound on the variational problem for $M_{\gamma}$ and were able to provide upper bounds that in dimensions $d\geq 5$ improved the till then best upper bounds on the optimal constant in the CLR inequality due to Lieb \cite{Lie80}.  
However, the choice of test functions in \cite[Appendix D]{HKRV23} seems quite arbitrary and the questions of existence of optimizers and the corresponding optimal value for $M_{\gamma}$ were left open. 

In the article at hand we prove that the variational problem \eqref{eq:primal} indeed has a solution, and we give a description of the optimizer in terms of the solution of another optimization problem, which makes an interesting connection to (a variant of) Hadamard's three lines lemma.

In order to state this connection, let us consider the set $\mathbb{H}(S)$ of holomorphic functions on the strip 
\begin{align*}
	S \coloneqq \{ x+\ii y \in \C: 0<y<1\}
\end{align*}
in the complex plane, and introduce the following class of spaces similar to Hardy spaces:

\begin{definition}[$\mathbb{H}^{p,q}(S)$ spaces] \label{def:pq-hardy}
Let $1\leq p, q\leq \infty$. We say that a holomorphic function $h:S \rightarrow \C$ belongs to the space $\mathbb{H}^{p,q}(S)$ if
\begin{align}
	\norm{h}_{\mathbb{H}^{p,q}(S)} \coloneqq	 \sup_{0<y<1} \inf_{\substack{f_y\in L^p(\RR), g_y\in L^q(\RR): \\ f_y+g_y  = h_y}} \left( \frac{1}{1-y}\norm{f_y}_{L^p(\R)} + \frac{1}{y} \norm{g_y}_{L^q(\R)} \right) <\infty, \label{eq:hardynorm}
\end{align}
where for $0<y<1$, $h_y:\R \to \C$ denotes the function\footnote{Note that we do not require the functions $f,g$ in the above decomposition of $h$ to be holomorphic.} 
\begin{align*}
    h_y(x) \coloneqq h(x+\ii y).
\end{align*}
Moreover, we say that a holomorphic function $h:S_- \rightarrow \C$ on the lower strip $S_- \coloneqq \{ x-\ii y \in \C: 0<y<1\}$ belongs to $\mathbb{H}^{p,q}(S_-)$ provided that the reflected function $z \in S\mapsto \overline{h(\overline{z})}$ belongs to $\mathbb{H}^{p,q}(S)$.
\end{definition}
We will show in Lemma~\ref{lem:existenceboundary} that such functions admit boundary values in certain $L^p$ spaces. More precisely, for any $h\in \mathbb{H}^{p,2}(S)$, there exists $h_0 \in L^p(\R)$ and $h_1 \in L^2(\R)$ such that
\begin{align*}
    \lim_{y\downarrow 0} \inner{h_y, \varphi} = \inner{h_0, \varphi} \quad \mbox{and}\quad \lim_{y\uparrow 1}\inner{h_y,\varphi} = \inner{h_1,\varphi},
\end{align*}
for any Schwartz function $\varphi \in \mathcal{S}(\R)$.

\subsection{Main result} The first main result of this paper can now be stated as follows.
\begin{theorem}[Three lines variational problem]\label{thm:threelines} 
Let $\gamma >2$ and $M_\gamma$ be defined via \eqref{eq:primal}. Then we have
\begin{align}
	M_\gamma =  16 \pi \frac{(\gamma-2)^{\gamma-2}}{\gamma^{\gamma+1}} \biggr(\max_{h \in \mathbb{H}^{\infty,2}(S_-)}  \frac{\norm{h_{-2/\gamma}}_{L^\infty(\R)}}{\norm{h_0}_{L^\infty(\R)}^{1-2/\gamma} \norm{h_{-1}}_{L^2(\R)}^{2/\gamma}}\biggr)^\gamma, \label{eq:3linesproblem}
\end{align}
where $h_0$ and $h_{-1}$ are the boundary values of $h$ in the sense described above.
Moreover, the maximizer of \eqref{eq:3linesproblem} exists and is unique up to the transformation $h_{\alpha,\beta,\omega}(z) = \beta h(z-\alpha)\mathrm{e}^{\ii\omega z}$ for $\alpha,\omega \in \R$ and $\beta \in \C \setminus \{0\}$. 
\end{theorem}

This three lines variational problem resembles the variational problem associated to the famous Hadamard three lines lemma. In the latter, one replaces the $L^2$ norm in the denominator by the $L^\infty$ norm; one can then apply the maximum principle to conclude that the constant function is the unique maximizer (up to the symmetry described above), which gives the statement of the Hadamard three lines lemma. For problem~\eqref{eq:3linesproblem}, however, the same strategy does not apply as the functions considered 
\begin{enumerate}[label=(\roman*)]
	\item need to be in $L^2$ along the line $\R+i$, in particular they may be unbounded,
	\item only extend to the boundary of the strip in a weak sense to be made precise later, and
	\item the maximum principle does not apply. 
\end{enumerate}
Nevertheless, it turns out that problem~\eqref{eq:3linesproblem} can also be solved in the sense that an explicit analytic formula for the optimizers can be given. This is the content of the next theorem.

\begin{theorem}[Optimizer] \label{thm:optimizer}
For $\gamma >2$, the unique optimizer (up to symmetries) of Problem~\eqref{eq:3linesproblem} is given by
\begin{align*}
 h(z) = B_\gamma(z) \mathrm{e}^{\theta_\gamma(z)},
\end{align*}
where $B_\gamma(z) = \frac{z-z_{\gamma}}{z-\overline{z}_{\gamma}}$ is the Blaschke factor (for the upper half plane) with zero at $z_{\gamma} = \ii(2-\frac{2}{\gamma}) \in \CC_+$, and the function $\theta_\gamma$ is defined via
\begin{align*}
	\theta_\gamma(z) \coloneqq \frac{1}{2\pi} \lim_{\epsilon \downarrow 0} \int_{|k| \geq \epsilon} \frac{g_\gamma(k) (\mathrm{e}^{\ii zk} - \ii zk)}{k(\cosh(2k)-1)} \mathrm{d}k, 
    \text{ with } 
    g_\gamma(k) \coloneqq \pi \left(2\mathrm{e}^{-(2-\frac{2}{\gamma})|k|}  + \mathrm{e}^{-\frac{2}{\gamma}|k|} - \mathrm{e}^{-(4-\frac{2}{\gamma})|k|}\right).
\end{align*}
\end{theorem}

From the above formula, one can easily evaluate the values of $M_\gamma$ to high precision. The numerical values\footnote{The numerical values were obtained using the standard numerical integration routines of MATLAB and independently cross-checked with Mathematica.} of $M_\gamma$ for some values of $\gamma$ are displayed in Table~\ref{tab:Mgamma}.
\begin{table}[h]
\begin{tabular}{|c|c|}
 \hline
 $\gamma$ & $M_\gamma$ \\
 \hline
 3 & 0.371185695  \\
 \hline
 4 & 0.098174770  \\
 \hline
 5 & 0.040698664 \\
 \hline
 6 & 0.020862684\\
 \hline
 7 & 0.012143294\\
 \hline
 8 & 0.007698202 \\
 \hline
 9 & 0.005190491 \\
 \hline
\end{tabular}
\vspace*{0.2cm}
\caption{$M_\gamma$ for different values of $\gamma$.}\label{tab:Mgamma}
\end{table}

\subsection{Applications to CLR and LT inequalities}
Our main motivation for studying the variational problem~\eqref{eq:primal} comes from its recently discovered connection to the CLR inequality \cite{HKRV21,HKRV23}. Remarkably, the same variational problem also appears in a recent upper bound for the optimal constant in the Lieb--Thirring (LT) inequality \cite{FHJN21}. Let us briefly recall these results and explore their consequences in connection with our main results.

We start with the CLR inequality, for which in \cite[Theorem 1.3]{HKRV23} the authors proved the following upper bound for the CLR constant:
\begin{theorem}[CLR bound] \label{thm:CLRbound} Let $d \in \N$ and $0<\sigma <d/2$. Then the best constant $L_{0,d,\sigma}$ in the Cwikel-Lieb-Rozenblum inequality for the fractional Schr\"odinger operator $(-\Delta)^\sigma + V$,
\begin{align}
    N_0((-\Delta)^\sigma + V) \leq L_{0,d,\sigma} \int_{\R^d} V_-(x)^{\frac{d}{2\sigma}} \mathrm{d} x, \label{eq:CLR}
\end{align}
satisfies
\begin{align}
    \frac{L_{0,d,\sigma}}{L_{0,d,\sigma}^{\rm cl}} \leq \frac{1}{4} \frac{\gamma^{\gamma+1}}{ (\gamma-2)^{\gamma-2}} M_\gamma, \quad \mbox{with $\gamma = \frac{d}{\sigma}$,} \label{eq:CLRbound}
\end{align}
where the semi-classical constant is given by $L_{0,d,\sigma}^{\rm cl} = \frac{|B_1|}{(2\pi)^d}$ and $N_0$ denotes the number of negative eigenvalues (counting multiplicity).
\end{theorem}
We can therefore combine the values obtained for $M_\gamma$ with the above result to update the best known upper bounds for the CLR constant.  The result is displayed in Table~\ref{tab:CLRbound}.
\begin{table}[ht]
\begin{center}
\begin{tabular}{ |c|c|c| } 
 \hline
 $d/\sigma$ & Optimal upper bound via VP & Best known so far  \\ 
 \hline
 3 & 7.51651  & 6.86924 \\
 \hline
 4 &  6.28319 & 6.03398 \\
 \hline
 5 & 5.88812 & 5.95405 \\
 \hline
 6 & 5.70334 & 5.77058\\
 \hline
 7 & 5.60029 & 5.67647\\
 \hline
 8 & 5.53645 & 5.63198\\
 \hline
 9 & 5.49398  & 5.62080\\
 \hline
 $d/\sigma \ra \infty$ & 5.34282 & -- \\
 \hline
\end{tabular}
\vspace*{0.2cm}
\caption{Comparison between previously best known upper bounds on the CLR constant and the best possible bounds via our solution of the variational  problem (VP) for $M_{\gamma}$ in \eqref{eq:primal}. The values in the third column of Table~\ref{tab:CLRbound} were extracted from \cite[Table 1]{HKRV23}.}
\end{center}
\label{tab:CLRbound}
\end{table}

\begin{remark}[Asymptotics of CLR upper bound]
    The asymptotic bound 
    \begin{align}
        \limsup_{d/\sigma \ra \infty} \frac{L_{0,d,\sigma}}{L_{0,d,\sigma}^{\rm cl}} \leq 5.342823
    \end{align}
    in Table~\ref{tab:CLRbound} is proved in Section~\ref{sec:asymptotics}. Moreover, one can show that in the opposite limit $\gamma = \frac{d}{\sigma} \downarrow 2$, the right hand side of \eqref{eq:CLRbound} has the asymptotic behaviour
    \begin{align}\label{eq:bound-asymptotic}
        \frac{1}{4} \frac{\gamma^{\gamma+1}}{ (\gamma-2)^{\gamma-2}} M_\gamma = \frac{2}{(\gamma-2)^{\gamma-1}} + \mathcal{O}(1),
    \end{align}
    where $\mathcal{O}(1)$ denotes a term bounded by a constant as $\gamma \downarrow 2$; this follows from the lower and upper bounds in \cite[Proposition 1.4]{HKRV23}, and we refer to that article for details.
    Note that the blowup as $\gamma\downarrow 2$ in \eqref{eq:bound-asymptotic} is expected, since the CLR inequality~\eqref{eq:CLR} does \emph{not} hold in the case $\frac{d}{\sigma}\leq 2$ \cite{LSW02,HHRV23}.
\end{remark}

We now turn to the Lieb--Thirring inequality. In \cite[Theorem 2]{FHJN21}, the authors proved the following upper bound for the LT inequality. 
\begin{theorem}[LT Bound]\label{thm:LTbound} For any $d\in \N$ and $\sigma >0$, the best constant $L_{1,d,\sigma}$ in the Lieb--Thirring inequality for the fractional Schr\"odinger operator $(-\Delta)^\sigma + V$,
\begin{align}
    \mathrm{Tr}[(-\Delta)^\sigma + V]_- \leq L_{1,d,\sigma} \int_{\R^d} V_-(x)^{1+\frac{d}{2\sigma}} \mathrm{d} x, \label{eq:LT}
\end{align}
satisfies
\begin{align}
    \frac{L_{1,d,\sigma}}{L_{1,d,\sigma}^{\rm cl}} \leq \frac{(d+2\sigma)^{2+\frac{d}{2\sigma}}}{d^{\frac{d}{2\sigma}} (2\sigma)^2}\mathcal{C}_{d,\sigma}, \quad \mbox{where $L_{1,d,\sigma}^{\rm cl} = \frac{2\sigma}{d+2\sigma} \frac{|B_1|}{(2\pi)^d}$,}
\end{align}
and
\begin{align}
    \mathcal{C}_{d,\sigma} = \frac{d}{2 \sigma} \inf_{\substack{f, \phi \in L^2(\R_+) \\ \norm{f}_{L^2(\R_+)} = 1}} \biggr\{ \biggr(\int_0^\infty \phi(t)^2 \mathrm{d}t\biggr)^{\frac{d}{2\sigma}} \int_0^\infty \frac{\bigr(1-\int_0^\infty \phi(s) f(ts) \mathrm{d}s\bigr)^2}{t^{1+\frac{d}{2\sigma}}} \mathrm{d} t \biggr\}.\label{eq:primalNAM}
\end{align}
\end{theorem}
At first glance, problem~\eqref{eq:primalNAM} looks different from problem~\eqref{eq:primal}. However, setting $m_1(t) \coloneqq \sqrt{2} f(t^2) t$ and $m_2(t) =\sqrt{2} \phi(t^{-2}) t^{-1}$, one can easily verify that 
\begin{align*}
   \norm{m_1}_{L^2(\R_+,\tfrac{\mathrm{d}s}{s})} = \norm{f}_{L^2(\R_+)}, 
   \quad 
   \norm{m_2}_{L^2(\R_+,\tfrac{\mathrm{d}s}{s})} = \norm{\phi}_{L^2(\R_+)}, 
\end{align*}
and 
\begin{align*}
	\int_0^\infty f(ts) \phi(s) \mathrm{d}s = \frac{(m_1 \star m_2)(t^{\frac12})}{t^{\frac12}},
\end{align*}
where the convolution is understood with respect to the multiplicative group, see \eqref{eq:convolution}. 
Thus, from the change of variables $t^{\frac12} \to t$ and a scaling argument (see Lemma~\ref{lem:scaling}) we find that
\begin{align}
    \mathcal{C}_{d,\sigma} = \frac{d}{\sigma} M_{2+\frac{d}{\sigma}}, \quad \mbox{where $\mathcal{C}_{d,\sigma}$ is defined in \eqref{eq:primalNAM} and $M_\gamma$ is defined in \eqref{eq:primal}.}
\end{align}
Consequently, we can also use the computed values for $M_\gamma$ in Table~\ref{tab:Mgamma} to update the best known upper bounds on the LT constant. 
\begin{corollary}[LT Bound Updated]\label{cor:improvedLTbound} For any $d\in \N$ and $\sigma >0$, the best constant in the Lieb--Thirring inequality for the fractional Schr\"odinger operator $(-\Delta)^\sigma + V$ satisfies
\begin{align}
    \frac{L_{1,d,\sigma}}{L_{1,d,\sigma}^{\rm cl}} \leq \frac{1}{4}\frac{\gamma^{\frac{\gamma+2}{2}}}{(\gamma-2)^{\frac{\gamma-4}{2}}}  M_{\gamma}, \quad\quad \mbox{with $\gamma = 2 + \frac{d}{\sigma}$,}
\end{align}
where $M_\gamma$ is defined in \eqref{eq:primal}. In particular, we have
\begin{align}
    \frac{L_{1,1,1}}{L^{\rm cl}_{1,1,1}} \leq 1.44655, 
    \quad
    \frac{L_{1,3,1/2}}{L^{\rm cl}_{1,3,1/2}} \leq 1.75177,  
    \quad
    \limsup_{\frac{d}{\sigma} \ra \infty} \frac{L_{1,d,\sigma}}{L^{\rm cl}_{1,d,\sigma}} \leq 1.96551,
    \quad \text{and} \quad 
    \limsup_{\frac{d}{\sigma} \downarrow 0} \frac{L_{1,d,\sigma}}{L_{1,d,\sigma}^{\rm cl}} = 1. \label{eq:marginalimprovedLT}
\end{align}
\end{corollary}

Let us now compare the values derived above with previous results. For the LT inequality, our bound $L_{1,1,1}/L_{1,1,1}^{\rm cl} \leq 1.44655$ improves only marginally over the previously best known bound $L_{1,1,1}/L_{1,d,1}^{\rm cl} \leq 1.45579$ derived in \cite{FHJN21}. This shows that the \emph{lower} bounds derived for $\mathcal{C}_{d,\sigma}$ and $M_\gamma$ respectively in \cite[Corollary 8]{FHJN21}  and \cite[Proposition 1.4]{HKRV23} were rather optimistic and the exact optimal values of \eqref{eq:primal} are in fact much closer to the \emph{upper} bounds obtained in these works. Moreover, from the \emph{induction in dimension} argument due to Laptev--Weidl \cite{LW00} and Hundertmark--Laptev--Weidl \cite{HLW00}, it is well-known that
\begin{align}
 \frac{L_{1,d,1}}{L_{1,d,1}^{\rm cl}} \leq \frac{L_{1,d',1}}{L_{1,d',1}^{\rm cl}} \quad \mbox{for any $d \geq d' \in \N$.} \label{eq:dimensionlifting}
\end{align}
We can therefore extend the upper bound in \eqref{eq:marginalimprovedLT} from the case $d=1=\sigma$ to all dimensions $d \geq 1$ with $\sigma=1$; in particular, our results marginally improve the best upper bounds for these cases as well. On the other hand, for the polyharmonic Schr\"odinger operator with $\sigma \neq 1$, the induction-in-dimension argument is not available and our bounds significantly improve over the best known bounds as $\frac{d}{\sigma}$ becomes large. For instance, for the ultra-relativistic operator in three dimensions the bound in \eqref{eq:marginalimprovedLT} only slightly improves over the previous bound $L_{1,3,1/2}/L_{1,3,1/2}^{\rm cl} = (1/0.826)^3 \approx 1.77443$ \cite{FHJN21} while the asymptotic bound in \eqref{eq:marginalimprovedLT} is about $38\%$ better than the asymptotic bound  
\begin{align*}
    \limsup_{\frac{d}{\sigma} \ra \infty} \frac{L_{1,d,\sigma}}{L_{1,d,\sigma}^{\rm cl}} = \mathrm{e} \approx 2.71828
\end{align*}
in \cite[Corollary 3]{FHJN21}.

For the CLR inequality the situation is similar. The new upper bounds presented in Table~\ref{tab:CLRbound} do not improve over the original upper bounds derived by Lieb \cite{Lie76} for $d=3,4$. For dimension $d\geq 5$, however, we obtain improvements over the values derived in \cite{HKRV23}; these improvements become significantly better as $d/\sigma$ grows. Moreover, in contrast with the LT inequality, where the case $d=\sigma =1$ gives the best bound for all higher dimensions, the upper bound for the CLR inequality is monotonically decreasing in $d/\sigma$; therefore, the upper bounds for large $d/\sigma$ presented here (see Table~\ref{tab:CLRbound}) provide significant improvements not only for the fractional case but also for the case $\sigma=1$.

To conclude our brief comparison with previous works, let us make a few remarks. 
First, both the LT and the CLR inequalities can be stated in a dual form (see, e.g., \cite[Equation 6]{FHJN21}). In particular, the upper bounds derived here can be translated into lower bounds for the dual constants, denoted by $K_{1,d,\sigma}$ for the LT inequality in \cite{FHJN21}. 
Second, in view of \cite[Theorem 1.7]{HKRV23} and \cite[Remark 7]{FHJN21}, the bounds derived here also apply to operator-valued Schr\"odinger operators. 
Third, in view of Theorems~\ref{thm:CLRbound} and \ref{thm:LTbound}, we expect the variational problem for $M_{\gamma}$ to also provide interesting upper bounds for the whole family of LT inequalities (i.e.\ sums of $p$\textsuperscript{th} moments of the negative eigenvalues with $p\geq 0$). This will be the content of future work.
Finally, let us refer to the review articles \cite{Fra21,Nam21,Sch22} and the books \cite{FLW23,EFGHW21} for further information on recent developments regarding  Lieb--Thirring and related inequalities, as well as the state of the art on the Lieb--Thirring conjecture. 

\subsection{Outline of the proofs of our main theorems}

We now outline the proof of Theorems~\ref{thm:threelines} and ~\ref{thm:optimizer}. First, based on a characterization of the space of Fourier transforms of $L^1(\R)$ functions, we derive an alternative formulation \eqref{eq:primalL1} of the variational problem \eqref{eq:primal}. This reformulation, via  its symmetries, allows us to obtain a convex formulation of the problem; we then apply the Fenchel-Rockafellar duality theorem to obtain a dual formulation of \eqref{eq:primalL1}, see Theorem~\ref{thm:scaleinvariantdual} below.

The next step of our proof brings in some tools from complex analysis to restate both the primal and dual variational problems in \eqref{eq:primalL1} and \eqref{eq:scaleinvariantdual} as minimization problems over the Hardy-type spaces $\mathbb{H}^{p,q}(S)$ on a strip in the complex plane introduced above. More precisely, we show that the domains of these problems are in one-to-one correspondence with the boundary values of functions in $\mathbb{H}^{p,q}(S)$ for suitable $p,q \in [1,\infty]$. From this correspondence, and the previous reformulations, we immediately obtain the three lines problem stated in Theorem~\ref{thm:threelines}.

The final step in our proof is to effectively solve the Euler-Lagrange (EL) equations associated to the three lines problem from Theorem~\ref{thm:threelines}. This is a quite challenging problem, owing to the fact that the EL equation (see eqs.~\eqref{eq:primaloptimizer} and \eqref{eq:ELholomorphic}) is non-local and non-linear in the sense that it relates the boundary values of a meromorphic function at the opposite ends of a strip in a non-linear way. To overcome these challenges, we implement the three following main ideas:
\begin{enumerate}[label=(\roman*)]
	\item Factorize the poles (and zeros) of the optimizer via a Blaschke product decomposition. For this, it turns out that the single Blaschke factor introduced in Theorem~\ref{thm:dualoptimizer} suffices.
	\item Linearize the EL equation by applying a logarithmic transformation. More precisely, we make the ansatz $h(z) = B_\gamma(z) \mathrm{e}^{\theta_\gamma(z)}$ for the optimizer, which allows us to re-state the non-linear EL equation~\eqref{eq:ELholomorphic} for $h$ as a linear equation for $\theta_\gamma$.
	\item Overcome the non-locality by using the fact that holomorphic functions on the strip can be viewed as the Fourier transform of rapidly decaying tempered distributions with complex argument.\footnote{This is actually the idea that led us to introduce the Hardy like spaces $\mathbb{H}^{p,q}(S)$.} In other words, we make the ansatz $\theta_\gamma(z) = \widehat{\eta}(z)$ for some tempered distribution $\eta$, which results in a linear local equation that can be formally solved pointwise.
\end{enumerate}
The rest of the proof then consists in showing that the formal guess obtained for the optimizer is in the right function space and indeed satisfies the EL equation.

\subsection{Overview of the article}
We end this section with a brief outline of the rest of the paper. In Section~\ref{sec:dual}, we carry out the first step of our proof and obtain a primal and dual reformulation of problem~\eqref{eq:primal} over the classical Lebesgue spaces $L^1(\R)$ and $L^\infty(\R)$. 
In Section~\ref{sec:holomorphic}, we study the boundary values of functions in $\mathbb{H}^{p,q}(S)$, and derive the reformulation of our variational problem over these spaces. 
The construction of the optimizer of \eqref{eq:3linesproblem} is carried out in Section~\ref{sec:optimisers}. 
In Section~\ref{sec:asymptotics}, we prove the asymptotic upper bounds for the CLR and LT inequalities stated in Table~\ref{tab:CLRbound} and Corollary~\ref{cor:improvedLTbound}. 
In Appendix~\ref{sec:app-Fourier}, we briefly discuss how the characterization of the space of Fourier transform of $L^1(\R)$ functions (Lemma~\ref{lem:fourier-L1} below) naturally appears in the maximal Fourier multiplier bound derived in \cite[Theorem 4.2]{HKRV23}, which is of independent interest. 
The proof of some technical lemmas used in Section~\ref{sec:holomorphic} are presented in Appendix~\ref{sec:app-hardy}. 
In Appendix~\ref{sec:app-primaloptimizer}, we prove existence of optimizers for problem~\eqref{eq:primal}, even though this is not explicitly needed for the proof of our main theorems (since we explicitly construct the optimizers).
\addtocontents{toc}{\protect\setcounter{tocdepth}{1}}
\section{The Dual problem} \label{sec:dual}
In this section, we show that Problem~\eqref{eq:primal} can be reformulated in $L^1(\R)$ with respect to standard Lebesgue measure. We then derive a dual formulation for this problem on the space  $L^\infty(\R)$. 

\subsection{Formulation over \texorpdfstring{$L^1(\RR)$}{L1(R)}}
It is convenient to transform \eqref{eq:primal} to $\RR$ by an exponential change of coordinates, yielding
\begin{align}
M_\gamma = \inf_{m_1,m_2 \in L^2(\R)} \biggr\{ \biggr(\norm{m_1}_{L^2(\R)} \norm{m_2}_{L^2(\R)} \biggr)^{\gamma-2} \norm{m_1\ast m_2 - \exp}_{L^2_\gamma(\RR)}^2 \biggr \}, \label{eq:primalR}
\end{align}
where $\exp: \R \rightarrow \R$ is the exponential function $k \mapsto \mathrm{e}^k$, $L^2(\R)$ now denotes the classical Lebesgue space of square integrable functions on $\R$, $m_1\ast m_2$ is the standard convolution
\begin{align*}
m_1 \ast m_2(k) = \int_\R m_1(k-u) m_2(u) \mathrm{d} u,
\end{align*}
and for $\gamma \in \R$, $L^2_\gamma$ will be used throughout this article to denote the exponentially weighted $L^2$ spaces\footnote{Notice that the weight is exponentially growing in one direction and exponentially decaying in the other. This comes from the exponential change of coordinates transforming \eqref{eq:primal} into \eqref{eq:primalR}.}
\begin{align}
 	L^2_\gamma(\RR) = L^2(\R; \mathrm{e}^{-\gamma k} \mathrm{d} k) = \biggr\{ m: \R \rightarrow \C : \norm{m}_{L^2_\gamma}^2 = \int_\R |m(k)|^2 \mathrm{e}^{-\gamma k} \mathrm{d} k < \infty \biggr\}. \label{eq:expL2space}
 \end{align}
This sets our problem in the more conventional framework of $L^p$ spaces on $\R$.
 
We use the following convention for the Fourier transform of a function $m\in L^1(\RR)$,
\begin{align*}
	\widehat{m}(k) =  \int_{\RR} m(x) \,\mathrm{e}^{-\ii k x}\,\mathrm{d}x,
\end{align*}
as well as its extension to the space of tempered distributions $\mathcal{S}'(\RR)$. With this convention, we have the following characterization of the Fourier transform of integrable functions.
\begin{lemma}[Fourier transform of integrable functions] \label{lem:fourier-L1} 
 Let $m \in \mathcal{S}'(\R)$. Then $\widecheck{m} \in L^1(\R)$ if and only if there exists $m_1, m_2 \in L^2(\R)$ with $m_1 \ast m_2 = m$. Moreover, we have the equality
\begin{align}
	\norm{\widecheck{m}}_{L^1(\R)} = \min \{ \norm{m_1}_{L^2} \norm{m_2}_{L^2} : m_1 \ast m_2 = m\}, \label{eq:variationalL1}
\end{align} 
where the minimum is attained.
\end{lemma}

\begin{remark*} The proof of Lemma~\ref{lem:fourier-L1} is a straightforward consequence of the convolution property of the Fourier transform and Plancherel's theorem. In fact, the first statement is known and can be found, e.g., in \cite[Theorem 1.6.3]{Rud62}. However, we could not find a reference for the variational characterization of $\norm{\widecheck{m}}_{L^1}$ stated in \eqref{eq:variationalL1}. Therefore, we present the simple proof below.
\end{remark*}

\begin{remark*}[Connection with maximal Fourier multipliers] 
	Lemma~\ref{lem:fourier-L1} also allows us to  improve the maximal Fourier mutliplier bound in \cite[Theorem 2.1, Theorem 4.2]{HKRV23} in a natural way. For the details, see Appendix~\ref{sec:app-Fourier}.
\end{remark*}

\begin{proof}[Proof of Lemma~\ref{lem:fourier-L1}] 
Suppose that $\widehat{m} \in L^1(\R)$. Then there exist measurable functions $|\widecheck{m}| : \R \rightarrow [0,\infty)$ and $\theta:\R \rightarrow [0,2\pi)$ such that $\widecheck{m}(x) = |\widecheck{m}|(x) \mathrm{e}^{\ii \theta(x)}$. In particular, if we define $(2\pi)^{\frac12} \widecheck{m_1} = (2\pi)^{\frac12} \widecheck{m_2} \coloneqq \sqrt{|\widecheck{m}|} \mathrm{e}^{\ii \theta(x)/2} \in L^2(\R)$, then $m_1\ast m_2 = m$ by the convolution property of the Fourier transform, and 
\begin{align*}
\norm{\widecheck{m}}_{L^1} = (2\pi) \norm{\widecheck{m_1}\widecheck{m_2}}_{L^1} = (2\pi) \norm{\widecheck{m_1}}_{L^2} \norm{\widecheck{m_2}}_{L^2} = \norm{m_1}_{L^2} \norm{m_2}_{L^2}
\end{align*}
by Plancherel's theorem. The converse implication and the inequality $\norm{\widecheck{m_1}} \leq \norm{m_1}_{L^2} \norm{m_2}_{L^2}$ follows from reversing the previous steps and using the Cauchy-Schwarz inequality. 
\end{proof}

The above characterization is the first key step of our analysis; it allows us to reformulate problem~\eqref{eq:primalR} in the following way:
\begin{lemma}[Reformulation on $L^1(\R)$]\label{lem:L1primal}
Let $\gamma>2$ and define the functional \[F_{\gamma}: L^1(\RR) \to \RR \cup\{+\infty\}, \quad  m\mapsto F_{\gamma}(m) \coloneqq \|\widehat{m} - \exp\|_{L^2_{\gamma}}^2.\] 
Then $F_{\gamma}$ is strictly convex and lower semi-continuous in $L^1(\RR)$ with domain 
\begin{align*}
	\dom F_{\gamma} = \left\{ m\in L^1(\RR): \widehat{m} - \exp \in L^2_{\gamma}(\RR) \right\}.
\end{align*}
Moreover, we have
\begin{align}
	M_\gamma = \inf_{m\in L^1(\R)} \norm{m}_{L^1}^{\gamma-2}  F_\gamma(m), \label{eq:primalL1}
\end{align}	
with $M_\gamma$ defined in equation~\eqref{eq:primalR}.
\end{lemma}

\begin{remark}[Notation]\label{rem:Egamma}
	In Sections \ref{sec:holomorphic} and \ref{sec:optimisers}, we will mostly be working with the functional
\begin{align}
    \mathcal{E}_\gamma(m) \coloneqq \norm{m}_{L^1}^{\gamma-2}  F_\gamma(m) = \norm{m}_{L^1(\R)}^{\gamma-2} \norm{\widehat{m} -\exp}_{L^2_\gamma}^2, \label{eq:Efunctional}
\end{align}
with domain
\begin{align}\label{eq:Edomain}
	\dom \mathcal{E}_\gamma = \dom F_{\gamma} = \{ m \in L^1(\R) : \widehat{m} - \exp \in L^2_{\gamma} \}.
\end{align}
To simplify notation, we write $\exp_{\alpha}: \RR \to \RR$, $x\mapsto \exp_{\alpha}(x):= \mathrm{e}^{\alpha x}$ for $\alpha \in \RR$.
\end{remark}

\begin{proof}[Proof of Lemma~\ref{lem:L1primal}]
	The strict convexity of $F_\gamma$ follows directly from the strict convexity of the norm in $L^2_\gamma$ squared. Moreover, the lower semi-continuity follows from Fatou's lemma by observing that convergence in $L^1(\R)$ implies uniform convergence of the Fourier transform. The reformulation of \eqref{eq:primalR} in \eqref{eq:primalL1} is immediate from Lemma~\ref{lem:fourier-L1} after interchanging the role of real and reciprocal (Fourier) space by considering $m$ as a function in $L^1(\R)$.
\end{proof}

The functional $F_\gamma$ has the following scaling property: for any $\alpha >0$ and $m \in L^1(\R)$, there holds 
\begin{align}
    F_\gamma(m_\alpha) = \alpha^{\gamma-2} F_\gamma(m) \quad\mbox{where}\quad m_\alpha(x) \coloneqq \alpha^{-1-\ii x} m(x) . \label{eq:scaling}
\end{align}
As a consequence, Problem~\eqref{eq:primalL1} can be reformulated in several ways.
\begin{lemma}[Scaling property]\label{lem:scaling} For any $p,q >0$ and $a>0$ we have
\begin{align*}
    \inf_{m \in L^1(\R)} \left\{ F_\gamma(m)^p + a \norm{m}_{L^1}^q \right\}= C(\gamma,p,q,a) M_\gamma^{\frac{pq}{(\gamma-2)p+q}},
\end{align*}
where $M_\gamma$ is defined in \eqref{eq:primalR} and
\begin{align*}
    C(\gamma,p,q,a) = \biggr(\frac{(\gamma -2)p}{q}\biggr)^{\frac{q}{(\gamma-2)p+q}} \biggr(\frac{q}{(\gamma-2)p} + 1\biggr) a^{\frac{(\gamma-2)p}{(\gamma-2)p+q}}.
\end{align*}
Moreover, we have
\begin{align}
    M_\gamma = \inf_{\substack{m\in L^1(\R)\\ \norm{m}_{L^1} \leq 1}}  F_\gamma(m) \label{eq:primalX2}.
\end{align}
In particular, if a minimizer of \eqref{eq:primal} exists, it is unique up to the 
transformation $m_\alpha(x) \coloneqq m(x)\alpha^{-\ii x-1}$, $\alpha >0$.
\end{lemma}

\begin{proof} From \eqref{eq:scaling} we have
\begin{align}
    \inf_{m \in L^1(\R)}\{ F_\gamma(m)^p + a \norm{m}_{L^1}^q \} &= \inf_{m \in L^1(\R)} \inf_{\alpha>0} \{F_\gamma(m_\alpha)^p + a \norm{m_\alpha}_{L^1}^q \} \nonumber\\
    &= \inf_{m \in L^1(\R)} \inf_{\alpha>0} \{\alpha^{(\gamma-2)p} F_\gamma(m)^p + \alpha^{-q} a \norm{m}_{L^1}^q\}. \label{eq:scaling-proof}
\end{align}
Minimizing the function 
$\alpha\mapsto f_m(\alpha) \coloneqq \alpha^{(\gamma-2)p} F_\gamma(m)^p + \alpha^{-q} a \norm{m}_{L^1}^q$ 
yields
\begin{align*}
    \min_{\alpha >0} f_m(\alpha) = \biggr(\frac{(\gamma -2)p}{q}\biggr)^{\frac{q}{(\gamma-2)p+q}} \biggr(\frac{q}{(\gamma-2)p} + 1\biggr) a^{\frac{(\gamma-2)p}{(\gamma-2)p+q}} \biggr(\norm{m}_{L^1}^{\gamma-2} F_\gamma(m)\biggr)^{\frac{pq}{(\gamma-2)p+q}},
\end{align*}
which together with \eqref{eq:scaling-proof} completes the proof. Equation \eqref{eq:primalX2} follows from similar arguments. Moreover, the uniqueness of the minimizer (provided it exists) follows from the strict convexity of $F_\gamma$ and the convexity of $\norm{\cdot}_{L^1}$.
\end{proof}

\subsection{Duality}
We now observe that eq.~\eqref{eq:primalX2} can be re-written as
\begin{align*}
    M_\gamma = \inf_{m \in L^1(\R)} \{ F_\gamma(m) + G(m) \},
\end{align*}
where $G:L^1(\RR) \to \{0, +\infty\}$ is the characteristic function (in the convex analysis terminology) of the unit ball in $L^1$, i.e.\
\begin{align*}
	G(m) = \begin{cases}
		0, & \text{if}\quad  \|m\|_{L^1(\RR)}\leq 1,\\
		+\infty, & \text{otherwise}.
	\end{cases}
\end{align*} 
The idea is now to apply the Fenchel-Rockafellar duality theorem to obtain a dual formulation of our problem. For this, we start with the following extension of Plancherel's theorem.

\begin{lemma}[Plancherel's theorem]\label{lem:parseval} Let $1\leq p \leq \infty$, then for any $g \in L^p(\R)$ with $\widehat{g} \in L^2_{-\gamma}(\R)$ and $m \in L^{\frac{p}{p-1}}(\R)$ with $\widehat{m} \in L^2_\gamma(\R)$, there holds
\begin{align}
 \left\langle g, m \right\rangle \coloneq \int_\R \overline{g(x)} m(x) \mathrm{d} x= \frac{1}{2\pi} \int_\R \overline{\widehat{g}(k)} \widehat{m}(k) \mathrm{d} k= \frac{1}{2\pi} \left\langle \widehat{g}, \widehat{m}\right\rangle. \label{eq:parseval}
\end{align}
\end{lemma}

\begin{proof} 
Let $\widehat{\phi} \in C_c^\infty(\R)$ be a real-valued function satisfying $\widehat{\phi}(x) = 1$ for $|x| \leq 1/2$ and $\widehat{\phi}(x) = 0$ for $|x|\geq 1$, and set $\widehat{\phi}_\epsilon(x) = \widehat{\phi}(\epsilon x)$. From the definition of the Fourier transform we have $\widehat{\phi}_\epsilon \widehat{g} = \widehat{\phi_\epsilon \ast g}$ for any $g \in \mathcal{S}'(\R)$. Moreover, if $g \in L^p(\R)$ for some $1\leq p \leq \infty$ with $\widehat{g} \in L^2_{-\gamma}(\RR)$  for some $\gamma \in \R$, then by dominated convergence and the approximate identity property of $\phi_\epsilon$, respectively, we have
\begin{align*}
    \widehat{\phi}_\epsilon \widehat{g}  \ra \widehat{g} \quad \mbox{strongly in $L^2_{-\gamma}$}\quad \mbox{and}\quad  \phi_\epsilon \ast g \rightarrow g \quad \mbox{strongly in $L^p(\R)$ (or weak-$\ast$ for $p=\infty$)} 
\end{align*}
as $\epsilon \downarrow 0$. From the assumptions on $\widehat{\phi}$, we see that $\widehat{\phi}_\epsilon = \widehat{\phi}_\epsilon \widehat{\phi}_{\epsilon/2}$ and $\widehat{\phi}_\epsilon \widehat{g} \in L^2(\R)$. Similar properties hold with $g$ replaced by $m$. Hence, using the above convergence and Plancherel's theorem for $L^2(\R)$ functions, we conclude that
\begin{align*}
    \frac{1}{2\pi} \inner{\widehat{g}, \widehat{m}} = \lim_{\epsilon \downarrow 0}\frac{1}{2\pi} \inner{\widehat{\phi}_\epsilon \widehat{g}, \widehat{\phi}_{\epsilon/2} \widehat{m}} = \lim_{\epsilon \downarrow 0} \inner{\phi_\epsilon \ast g, \phi_{\epsilon/2} \ast m} = \inner{g,m}.
\end{align*}
\end{proof}

\begin{remark}
	Recall that the Fourier transform of a function $g\in L^{p}(\RR)$, $2<p\leq  \infty$ has to be understood in the sense of tempered distributions. Since the Fourier transform of a tempered distribution is itself a tempered distribution, the condition $\widehat{g} \in L^2_{-\gamma}(\RR)$ means that $\widehat{g}$ can be identified with a function that lies in $L^2_{-\gamma}(\RR)$, which puts a one-sided (exponential) decay constraint on the Fourier transform $\widehat{g}$. This allows us to define certain integrals of $\widehat{g}$ against exponentially growing functions. 
\end{remark}

\begin{lemma}\label{lem:explimit}
Let $\gamma>2$ and $g \in L^\infty(\R)$ with $\widehat{g} \in L^2_{-\gamma}(\R)$. Then the limit
\begin{align}
 \lim_{\epsilon \downarrow 0} \int_\R \phi(\epsilon k) \widehat{g}(k) \mathrm{e}^k \mathrm{d} k ,\quad \phi \in C^\infty_c(\R), \quad \phi(0) = 1, \label{eq:explimit}
\end{align}
exists and is independent of $\phi$. We can therefore define 
\begin{align}\label{eq:exp-duality}
	\left\langle \exp, \widehat{g} \right\rangle = \int_{\RR} \widehat{g}(k) \mathrm{e}^k \mathrm{d} k \coloneqq \lim_{\epsilon \downarrow 0} \int_\R \phi(\epsilon k) \widehat{g}(k) \mathrm{e}^k \mathrm{d} k 
\end{align}
for any $\phi \in C^\infty_c(\R)$ such that $\phi(0) = 1$. 
\end{lemma}
\begin{proof}
We decompose $\mathrm{e}^k = p(k) + q(k)$, where $p(k) \coloneqq \mathrm{e}^{-|k|}$ and $q(k) \coloneqq \left( \mathrm{e}^{k} - \mathrm{e}^{-k} \right)\1_{(0,\infty)}(k)$. Note that $\widehat{p}(x) = \frac{2}{1+x^2} \in L^1(\RR)$ and $q \in L^2_\gamma$ for $\gamma>2$. 

Let $\phi_\epsilon(k) \coloneqq \phi(\epsilon k)$ for $k\in\RR$, then $\widehat{\varphi_{\epsilon}}$ is an approximate identity in $L^1(\RR)$, in particular $\widehat{\phi_\epsilon} \ast \widehat{p} \ra \widehat{p}$ in $L^1(\R)$, and by the classical Parseval identity, we have 
\begin{align*}
	\int_{\RR} \varphi_{\epsilon}(k) p(k) \widehat{g}(k)\,\mathrm{d}k 
	= \frac{1}{2\pi} \int_{\RR} \left(\widehat{\phi}_\epsilon \ast \widehat{p}\right)(x) g(x) \mathrm{d}x 
	\stackrel{\epsilon\to 0}{\longrightarrow} \frac{1}{2\pi} \int_{\RR} \widehat{p}(x) g(x)\,\mathrm{d}x = \frac{1}{\pi} \int_{\RR} \frac{g(x)}{1+x^2}\,\mathrm{d}x. 
\end{align*}
The limit is independent of the choice of $\varphi$ and, since $g\in L^{\infty}(\RR)$, it is finite. 
Moreover, $\phi_\epsilon q \ra q$ in $L^2_{\gamma}$ as $\epsilon \downarrow 0$ by dominated convergence, from which it follows that 
\begin{align*}
	\int_{\RR} \varphi_{\epsilon}(k) q(k) \widehat{g}(k) \stackrel{\epsilon\to 0}{\longrightarrow} \int_{\RR} q(k) \widehat{g}(k) \,\mathrm{d}k,
\end{align*}
where the limit is finite since $q\in L^2_{\gamma}$ and $\widehat{g}\in L^{2}_{-\gamma}$, and independent of $\varphi$ (as long as $\varphi(0) = 1$). 
\end{proof}

We can now compute the dual of $F_\gamma$. 

\begin{proposition}[Dual of $F_\gamma$] \label{lem:dualFgamma} The Fenchel conjugate $F_\gamma^\ast : L^\infty(\R) \rightarrow  \R \cup \{+\infty\}$ of $F_\gamma$ is given by
\begin{align*}
F_\gamma^\ast(g) &= \sup_{m\in \dom F_{\gamma}} \left( \Re \int_{\RR} \overline{g(x)} m(x) \,\mathrm{d}x - F_{\gamma}(m)\right) \\
&= \begin{dcases} 
\frac{1}{16 \pi^2} \norm{\widehat{g}}_{L^2_{-\gamma}}^2 + \frac{1}{2\pi} \mathrm{Re} \int_\R \widehat{g}(k)  \mathrm{e}^k \mathrm{d} k, \quad &\mbox{if $\widehat{g} \in L^2_{-\gamma}$,} \\
+\infty \quad&\mbox{otherwise,} \end{dcases}
\end{align*}
where $\int_\R \widehat{g}(k) \mathrm{e}^k \mathrm{d}k$ is understood in the sense of \eqref{eq:exp-duality}.
\end{proposition}

\begin{proof} 
Note that $\dom F_{\gamma} \neq \emptyset$, since $m(x) = \pi^{-1} (1+x^2)^{-1}$ is integrable with $\widehat{m}(k) = \mathrm{e}^{-|k|}$, so $m(k)-\mathrm{e}^k \in L^2_{\gamma}$ for any $\gamma>2$ (see the proof of Lemma~\ref{lem:explimit}).

To compute the dual, we fix some $m_0 \in \dom F_\gamma$ and split $m \in \dom F_{\gamma}$ in $m=m_0 + h$ with $h\in L^1(\RR)$ such that $\widehat{h} \in L^2_{\gamma}$. Then 
\begin{align*}
	F_\gamma^\ast(g) 
	&= \Re \int_{\RR} \overline{g(x)}m_0(x)\,\mathrm{d}x - F_{\gamma}(m_0) + \sup_{h\in L^1: \widehat{h}\in L^2_{\gamma}} \left( \Re \int_{\RR} \overline{g(x)}h(x)\,\mathrm{d}x + F_{\gamma}(m_0) - F_{\gamma}(m_0+h) \right).
\end{align*}
Note that $|\re \langle g, m_0 \rangle | \leq \|g\|_{L^{\infty}} \|m_0\|_{L^1} < \infty$ and $F_{\gamma}(m_0) < \infty$ since $m_0 \in \dom F_{\gamma}$. Hence, $F_{\gamma}^*$ is finite if and only if 
\begin{align*}
	&\sup_{h\in L^1: \widehat{h}\in L^2_{\gamma}} \left( \Re \langle g, h \rangle - 2 \re \left\langle (\widehat{m_0} - \exp) \exp_{-\gamma}, \widehat{h} \right\rangle - \|\widehat{h}\|_{L^2_{\gamma}}^2 \right) \\
	&= \sup_{h\in L^1: \widehat{h}\in L^2_{\gamma}} \sup_{\alpha\in\CC} \left( \Re \alpha \langle g, h \rangle - 2 \re \alpha \left\langle (\widehat{m_0} - \exp) \exp_{-\gamma}, \widehat{h} \right\rangle - |\alpha|^2 \|\widehat{h}\|_{L^2_{\gamma}}^2 \right)\\
	&= \sup_{h\in L^1: \widehat{h}\in L^2_{\gamma}} \frac{1}{4} \frac{\left
	|\langle g,h\rangle - 2 \left\langle (\widehat{m_0} - \exp) \exp_{-\gamma}, \widehat{h} \right\rangle \right|^2}{\|\widehat{h}\|_{L^2_{\gamma}}^2} < \infty.
\end{align*}
Since $C_c^\infty(\R) \subset \{ \widehat{h} \in L^2_\gamma : h \in L^1\}$ and the former is dense in $L^2_\gamma$, the above quotient is finite if and only if the linear functional $\ell: C_{c}^{\infty}(\R) \to \RR$ given by 
\begin{align*}
	\widehat{h} \mapsto \ell(\widehat{h}) \coloneqq \langle g,h\rangle - 2 \left\langle (\widehat{m_0} - \exp) \exp_{-\gamma}, \widehat{h} \right\rangle
\end{align*}
extends to a bounded linear functional on $L^2_{\gamma}(\RR)$. In this case, the Riesz representation theorem implies that there exists a unique $G \in L^2_{-\gamma}(\RR)$ such that $\ell(\widehat{h}) = \langle G, \widehat{h} \rangle$ for all $\widehat{h}\in L^2_{\gamma}(\R)$. In particular, since $(\widehat{m_0} - \exp) \exp_{-\gamma} \in L^2_{-\gamma}$ (recall that $\widehat{m_0} - \exp \in L^2_{\gamma}$), it follows that 
\begin{align*}
	\langle g, h \rangle = \left\langle G + 2 (\widehat{m_0} - \exp) \exp_{-\gamma}), \widehat{h} \right\rangle \quad \text{for all } \widehat{h} \in L^2_{\gamma},
\end{align*}
from which we conclude that $\widehat{g} \in L^2_{-\gamma}(\RR)$. We may therefore apply Parseval's identity (Lemma~\ref{lem:parseval}) to obtain 
\begin{align*}
	\sup_{h\in L^1: \widehat{h}\in L^2_{\gamma}} \frac{1}{4} \frac{\left
	|\langle g,h\rangle - 2 \left\langle (\widehat{m_0} - \exp) \exp_{-\gamma}, \widehat{h} \right\rangle \right|^2}{\|\widehat{h}\|_{L^2_{\gamma}}^2}
	&= \sup_{h\in L^1: \widehat{h}\in L^2_{\gamma}} \frac{1}{4} \frac{\left|\left\langle \frac{1}{2\pi} \widehat{g} - 2 (\widehat{m_0} - \exp) \exp_{-\gamma}, \widehat{h} \right\rangle\right|^2}{\|\widehat{h}\|_{L^2_{\gamma}}^2} \\
	&= \frac{1}{4} \left\| \frac{1}{2\pi} \widehat{g} - 2 (\widehat{m_0} - \exp) \exp_{-\gamma} \right\|_{L^2_{-\gamma}}^2.
\end{align*}
It follows that 
\begin{align}
	F^*_{\gamma}(g) 
	&= \re \langle g, m_0\rangle - \|\widehat{m_0}-\exp\|_{L^2_{\gamma}}^2 +  \left\| \frac{1}{4\pi} \widehat{g} - (\widehat{m_0} - \exp) \exp_{-\gamma} \right\|_{L^2_{-\gamma}}^2 \nonumber \\
	&= \frac{1}{(4\pi)^2} \|\widehat{g}\|_{L^2_{\gamma}}^2 + \re \langle g, m_0\rangle - \frac{1}{2\pi} \re \inner{\widehat{g}, \widehat{m_0} - \exp}. \label{eq:someid}
\end{align}
To complete the proof, we observe that by dominated convergence, for any $\phi \in C_c^\infty(\R)$ with $\phi(0) = 1$, setting $\phi_\epsilon \coloneqq \phi(\epsilon\cdot)$, there holds  
\begin{align*}
	&\frac{1}{2\pi} \Re \int_{\RR} \overline{\widehat{g}(k)} (\widehat{m}_0(k) - \mathrm{e}^{k}) \,\mathrm{d}k 
	= \lim_{\epsilon\downarrow 0} \frac{1}{2\pi} \Re \int_{\RR} \overline{\widehat{g}(k)} \phi_{\epsilon}(k)(\widehat{m}_0(k) - \mathrm{e}^{k}) \,\mathrm{d}k \\
	&= \lim_{\epsilon\downarrow 0} \Re \int_{\RR} \overline{g(x)} \widecheck{\phi}_{\epsilon}*m_0(x)\,\mathrm{d}x - \lim_{\epsilon\downarrow 0} \frac{1}{2\pi} \int_{\RR} \overline{\widehat{g}(k)} \phi_{\epsilon}(k) \mathrm{e}^k \,\mathrm{d}k \\
	&= \Re \int_{\RR} \overline{g(x)} m_0(x)\,\mathrm{d}x - \frac{1}{2\pi} \Re \int_{\RR} \overline{\widehat{g}(k)} \mathrm{e}^{k}\,\mathrm{d}k.
\end{align*}
With this and \eqref{eq:someid}, we can conclude that 
\begin{align*}
	F^*_{\gamma}(g) =
		\frac{1}{(4\pi)^2} \|\widehat{g}\|_{L^2_{\gamma}}^2 + \frac{1}{2\pi} \re \left\langle \widehat{g}, \exp \right\rangle, \quad \text{if } \widehat{g} \in L^2_{-\gamma}(\RR),
\end{align*}
and $F^*_{\gamma}(g) = + \infty$ otherwise.
\end{proof}

From the above lemma and the Fenchel-Rockafellar duality theorem we obtain
\begin{lemma}[Duality]\label{lem:dualproblem} The strong duality
\begin{align}
	M_\gamma = \inf_{\substack{m\in L^1(\R)\\ \norm{m}_{L^1} \leq 1}}  F_\gamma(m) &= - \min_{\substack{g \in L^\infty \\ \widehat{g} \in L^2_{-\gamma}}} \biggr\{\norm{g}_{L^\infty} + \frac{1}{16\pi^2} \norm{\widehat{g}}_{L^2_{-\gamma}}^2 + \frac{1}{2\pi} \mathrm{Re}\int_{\RR} \widehat{g}(k) \mathrm{e}^k \,\mathrm{d}k \biggr\}\label{eq:dual}
\end{align}
holds. Moreover, there exists a unique minimizer for the problem on the right-hand side.
\end{lemma}

\begin{proof} By the Fenchel-Rockafellar duality theorem \cite[Theorem 31.1]{Roc72}, we have
\begin{align*} 
	\inf_{\substack{m\in L^1(\R)\\ \norm{m}_{L^1} \leq 1}}  F_\gamma(m) 
	= \inf_{m \in L^1(\R)} \{F_\gamma(m) + G(m) \}
	= - \min_{g \in L^\infty(\R)} \{ F_\gamma^\ast (g) + G^\ast(-g) \}.
\end{align*}
where the minimizer 
of the dual problem exists provided that both $F_\gamma$ and $G$ are convex and lower semi-continuous in $L^1(\R)$, and that there exists at least one point of continuity of $G$ in the domain of $F_\gamma$. 
The lower-semicontinuity and convexity can be directly checked for $G$ and are proved in Proposition~\ref{lem:dualFgamma} for $F_\gamma$. 
To verify the last condition, note that $m_\alpha(x) = m(x)\alpha^{-\ii x-1} \in \dom F_\gamma$ for any $\alpha >0$ and $m \in  \rm{dom}\, F_\gamma$. So for $\alpha$ big enough, we have $\norm{m_\alpha}_{L^1} < 1$, hence $m_\alpha$ is a continuity point of $G$ in the domain of $F_\gamma$. 
The uniqueness of the minimizer follows from the strict convexity of $\norm{\cdot}_{L^2_{-\gamma}}^2$. 
It remains to calculate 
\begin{align*}
	G^*(g) &= \sup_{m\in L^1} \left( \Re \int_{\RR}  \overline{g(x)}m(x)\,\mathrm{d}x - G(m) \right) 
	= \sup_{\substack{m\in L^1 \\ \|m\|_{L^1}\leq 1}}  \Re \int_{\RR}  \overline{g(x)} m(x)\,\mathrm{d}x 
	= \|g\|_{L^{\infty}},
\end{align*}
which completes the proof.
\end{proof}

As for the primal problem, we can derive a scale invariant version of the dual problem. 
\begin{theorem}[Scale invariant dual problem]\label{thm:scaleinvariantdual} Let $M_\gamma$ be defined as in \eqref{eq:primal}, then we have
\begin{align}
    M_\gamma = \frac{4(\gamma-2)^{\gamma-2}}{(2\pi)^{\gamma-2} \gamma^\gamma} \biggr(\max_{\substack{g \in L^\infty(\R) \\ \widehat{g} \in L^2_{-\gamma}}} \mathcal{E}^\ast_\gamma(g)\biggr)^{\gamma},\quad\mbox{where} \quad \mathcal{E}^\ast_\gamma(g) \coloneqq  \frac { \re \int_\R \widehat{g}(k) \mathrm{e}^k \mathrm{d} k }{\norm{\widehat{g}}_{L^2_{-\gamma}}^{\frac{2}{\gamma}} \norm{g}_{L^\infty}^{1-\frac{2}{\gamma}}}. \label{eq:scaleinvariantdual}
\end{align}
Moreover, the maximizer exists and is unique up to re-scaling and translating in Fourier space, i.e., up to the transformation $g_{\alpha,\beta}(x) = \beta \alpha^{\ii x} g(x)$ with $\beta >0$ and $\alpha >0$.
\end{theorem}

\begin{proof} First, note that up to multiplying $g$ by a phase $\mathrm{e}^{\ii\theta}$, $\theta \in \R$, the term $\mathrm{Re} \int_\R \widehat{g}(k) \mathrm{e}^k \mathrm{d}k$ appearing in the dual problem~\eqref{eq:dual} can be replaced by $-|\int_\R \widehat{g}(k) \mathrm{e}^k \mathrm{d}k|$. Next, by defining $g_{\alpha,\beta}(x) = \beta g(x) \alpha^{\ii x}$ and noticing that $\widehat{g}_{\alpha,\beta}(k) = \beta \widehat{g}(k-\log \alpha)$, we find 
\begin{align*}
    M_\gamma &= - \inf_{g} \inf_{\alpha>0,\beta >0} \biggr\{ \norm{g_{\alpha,\beta}} +  \frac{1}{16 \pi^2} \norm{\widehat{g}_{\alpha,\beta}}_{L^2_{-\gamma}}^2 - \frac{1}{2\pi} \biggr|\int_\R \widehat{g}_{\alpha,\beta}(k) \mathrm{e}^k \mathrm{d}k\biggr| \biggr\} \\
    &= -\inf_{g} \inf_{\alpha>0, \beta>0} \biggr\{\beta \norm{g}_{L^\infty} + \beta^2 \alpha^\gamma \frac{1}{16\pi^2} \norm{\widehat{g}}_{L^2_{-\gamma}}^2 - \beta \alpha \frac{1}{2\pi} \biggr|\int_\R \widehat{g}(k) \mathrm{e}^k \mathrm{d}k\biggr| \biggr\}.
\end{align*}
Hence our task reduces to finding the minimizer of the function 
\begin{align*}
	f(\alpha,\beta) = \beta c_1 + \beta^2 \alpha^\gamma c_2  - \beta \alpha c_3\quad\mbox{ with $c_1,c_2 > 0$ and $c_3\geq 0$.}
\end{align*}
This can be done by finding the (unique) critical point of $f$, which is given by
\begin{align*}
\alpha_0 = \frac{c_1}{c_3} \frac{\gamma}{\gamma-2}, \quad\quad \beta_0 = \frac{c_3^\gamma}{c_1^{\gamma-1} c_2} .
\end{align*}
Eq.~\eqref{eq:scaleinvariantdual} then follows by substituting back the values of $c_1, c_2, c_3$ and evaluating $f$ at $(\alpha_0,\beta_0)$. \end{proof}

\section{Reformulation in the complex plane}
\label{sec:holomorphic}
In this section we first derive properties of the mixed Hardy-type spaces of holomorphic functions on the strip introduced in Definition~\ref{def:pq-hardy} in the introduction. It turns out that their (non-tangential) limits along the real axis correspond to the domains of the primal and dual problems. This allows us to interpret our problem as a variant of the variational problem associated with the classical Hadamard three lines lemma, as stated in Theorem~\ref{thm:threelines}. 
The connection with holomorphic functions will also be useful to characterize the primal and dual optimizers in Section~\ref{sec:optimisers}.

\subsection{Mixed Hardy-like spaces on the strip}
The main result of this section is that the space $\{f\in L^p(\R) : \widehat{f} \in L^2_2(\RR)\}$ can be identified with the boundary values of functions in $\mathbb{H}^{p,2}(S)$.

To prove this, we need the following lemma.

\begin{lemma}[Existence of boundary values in $\mathbb{H}^{p,2}(S)$]\label{lem:existenceboundary} Let $h\in \mathbb{H}^{p,2}(S)$, then there exist $h_0 \in L^p(\R)$ with $\widehat{h_0} \in L^2_2(\RR)$ and $h_1 \in L^2(\R)$ such that
\begin{align}
 \lim_{y\downarrow 0} \inner{h_y, \phi} = \inner{h_0, \phi} \quad\mbox{and}\quad \lim_{y\uparrow 1} \inner{h_y,\phi} = \inner{h_1,\phi} \label{eq:boundaryvalues}
\end{align}
for any Schwartz function $\phi \in \mathcal{S}(\R)$. Moreover, for any $\phi \in C_c^\infty(\R)$, there holds
\begin{align}
    \inner{h_y, \widehat{\phi}_{1-y}} = \inner{h_0, \widehat{\phi}_1} \quad \mbox{for any $0 \leq y \leq 1$,} \label{eq:cauchytheorem}
\end{align}
where $\widehat{\phi}_y(k) = \widehat{\phi}(k+\ii y) = \int_\R \mathrm{e}^{-\ii x(k+\ii y)} \phi(x) \mathrm{d} x$. In particular, if $h_0 =0$, then $h = 0$.
\end{lemma}

The proof of the above lemma is an adaptation of standard arguments used in the theory of Hardy spaces (see, e.g., \cite[Lemma 11.3]{Mas09} and \cite{Koo98,Gar06}). For convenience of the reader, we present the details in Appendix~\ref{sec:app-hardy}.

\begin{theorem}[Boundary values in $\mathbb{H}^{p,2}(S)$]\label{thm:holomorphicextension} Let $1\leq p \leq \infty$. For any $v\in L^p(\R)$ with $\widehat{v} \in L^2_{2}(\RR)$ there exists a unique function $h\in \mathbb{H}^{p,2}(S)$ such that $h_0 = v$. Moreover, there holds $\widehat{h}_1(k) = \widehat{v}(k) \mathrm{e}^{-k}$ for almost every $k\in\RR$. 

Conversely, for any $h \in \mathbb{H}^{p,2}(S)$ we have $h_0\in L^p(\R)$ and $\widehat{h_0} \in L^2_{2}(\RR)$.
\end{theorem}

\begin{proof} 
Let $v\in L^p(\R)$ with $\widehat{v} \in L^2_{2}(\RR)$. We have to show that $v$ has a unique holomorphic extension to the strip $S$ with $h_1 \in L^2(\RR)$. In fact, the uniqueness of such an extension is immediate from the last statement in Lemma~\ref{lem:existenceboundary}, so we just need to prove its existence. 
For this, the idea is to consider $\widehat{h_y}(k):= \widehat{v}(k) \,\mathrm{e}^{-y k}$ for $k\in\RR$ and $y\in(0,1)$, and construct $h\in \mathbb{H}^{p,2}(S)$ via Fourier inversion $\frac{1}{2\pi} \int_{\RR} \widehat{v}(k) \,\mathrm{e}^{\ii k z}\,\mathrm{d}k$. 
Since $\widehat{v} \in L^2_{2}(\R)$, some care has to be taken. To this end, we split the function $\mathrm{e}^{\mathrm{i} k z}$ into two parts
\begin{align*}
	\mathrm{e}^{\mathrm{i} k z} = \mathrm{e}^{\mathrm{i} k x} \mathrm{e}^{-k y} =  \mathrm{e}^{\mathrm{i} k x} \left(\widehat{p_y}(k) + q_y(k) \right)
\end{align*}
for some suitably chosen functions $p_y,q_y: \RR \to \CC$, and set 
\begin{align*}
	\widetilde{h}(x+\ii y):= (v*p_y)(x) + \frac{1}{2\pi} \int_{\RR} q_y(k) \widehat{v}(k) \mathrm{e}^{\mathrm{i} k x} \,\mathrm{d}k, \quad y\in (0,1).
\end{align*}
Since $v\in L^p(\RR)$, if $p_y \in L^1(\RR)$, Young's inequality implies that $\|v*p_y\|_{L^p} \leq \|v\|_{L^p} \|p_y\|_{L^1}$. On the other hand, since $\widehat{v}\in L^2_2(\RR)$, the second expression is a well-defined $L^2$-function if $q_y \exp \in L^{\infty}(\RR)$. 
In particular, we obtain 
\begin{align*}
	\|\widetilde{h}\|_{\mathbb{H}^{p,2}} \leq \sup_{0<y<1} \left( \frac{1}{1-y} \|v\|_{L^p} \|p_y\|_{L^1} + \frac{1}{y} \|\widehat{v}\|_{L^2_2} \|q_y \exp\|_{L^{\infty}}\right) .
\end{align*}
Hence, in order to obtain $\|\widetilde{h}\|_{\mathbb{H}^{p,2}} < +\infty$, we have to be able to choose $p_y, q_y$ in such a way that
\begin{align}\label{eq:conditions}
	\|p_y\|_{L^1} \lesssim 1-y \quad \text{and} \quad \|q_y \exp\|_{L^{\infty}} \lesssim y
\end{align}
for $y\in (0,1)$. A possible choice is given by 
\begin{align*}
q_y(k) \coloneqq 
\begin{cases}  
	(\mathrm{e}^{-yk}-\mathrm{e}^{yk}) 1_{\{k\leq 0\}} , &\text{for } y<1/2,\\
	\mathrm{e}^{-yk} 1_{\{k\leq 0\}} + \mathrm{e}^{-k} 1_{\{k >0\}}, & \text{for } y\geq 1/2,
\end{cases} 
\end{align*}
and 
\begin{align*}
\widehat{p_y}(k) \coloneqq 
\begin{cases}
	\mathrm{e}^{-y|k|}, &\text{for } y<1/2,\\
	(\mathrm{e}^{-yk} - \mathrm{e}^{-k})1_{\{k\geq 0\}}, & \text{for } y\geq 1/2.
\end{cases}
\end{align*}
Then $\mathrm{e}^{\ii xk} (\widehat{p}_y(k) + q_y(k)) = \mathrm{e}^{\ii zk}$ for any $z\in S$, and  
\begin{align*}
 p_y(x) = \begin{dcases} \frac{1}{\pi} \frac{y}{y^2+x^2}, \quad &\mbox{for $y<1/2$,}\\
 \frac{1}{2\pi} \frac{1-y}{(1-\ii x)(y-\ii x)}, \quad &\mbox{for $y\geq 1/2$.} \end{dcases}
\end{align*}
In particular, for any $0<y<\frac{1}{2}$ there holds $\norm{p_y}_{L^1(\R)} = 1$ and $\norm{q_y \exp}_{L^\infty(\R)} \lesssim y$, while for $\frac{1}{2} \leq y <1$ we have $\norm{p_y}_{L^1(\R)} \lesssim 1-y$ and $\norm{q_y \exp}_{L^\infty(\R)} \lesssim 1$. Therefore, conditions \eqref{eq:conditions} are fulfilled. 

Note that by construction, the function $\widetilde{h}$ attains the function $v$ as boundary value as $y \downarrow 0$ in a distributional sense. Indeed, for any Schwartz function $\phi$ we find that
\begin{align*}
	\inner{\widetilde{h}_y, \phi} &=  \inner{p_y \ast v, \phi} + \frac{1}{2\pi} \int_\R \overline{q_y \widehat{v}(k)} \widehat{\phi}(k) \mathrm{d}k 
	= \frac{1}{2\pi} \inner{\widehat{v} (\widehat{p_y} + q_y), \widehat{\phi}} 
	= \frac{1}{2\pi} \inner{\widehat{v} \mathrm{e}^{-y (\cdot)}, \widehat{\phi}} \\
	&\stackrel{y\downarrow0}{\longrightarrow} \frac{1}{2\pi} \inner{\widehat{v}, \widehat{\phi}} = \inner{v, \phi}.
\end{align*}

It remains to show that $\widetilde{h}:S \to \CC$ is a holomorphic function. For this, we perform a different decomposition $\mathrm{e}^{\ii zk} = \overline{\widehat{r}_z(k)} + \overline{s_z(k)}$ with 
\begin{align*}
	\widehat{r}_z(k) \coloneqq \mathrm{e}^{-\ii \overline{z} |k|} \quad\mbox{and}\quad s_z(k) \coloneqq \left( \mathrm{e}^{-\ii \overline{z} k} - \mathrm{e}^{\ii \overline{z} k} \right) \1_{\{k\leq 0\}}.
\end{align*}
Then $r_z(\xi) = \frac{1}{\pi} \frac{\ii \overline{z}}{\xi^2-\overline{z}^2}$ for $\xi \in \RR$, in particular $r_z \in L^1(\RR) \cap L^{\infty}(\RR)$ for any $z\in S$. In view of Lemma~\ref{lem:parseval} we then define the function 
\begin{align*}
	h(z) := \int_\R \overline{r_z(\xi)} v(\xi) \mathrm{d}\xi + \frac{1}{2\pi} \int_\R \overline{s_z(k)} \widehat{v}(k) \mathrm{d}k.
\end{align*}
Note that the second term is well-defined because $\widehat{v}\in L^2_{2}(\RR)$ and 
\begin{align*}
	\int_{\RR} |s_z(k)|^2 \mathrm{e}^{2k}\,\mathrm{d}k 
	\lesssim \int_{-\infty}^0 \mathrm{e}^{2(1-y)k}\,\mathrm{d}k + \int_{-\infty}^0 \mathrm{e}^{2(1+y)k}\,\mathrm{d}k < \infty
\end{align*}
for any $0<y<1$, hence $s_z \in L^2_{-2}(\RR)$ for any $z\in S$.

Since $z\mapsto \overline{\widehat{r_z}(k)}, \overline{s_z(k)}$ are holomorphic functions on $S$ for any $k\in\RR$, the function $h$ is holomorphic on $S$ by Morera's and Fubini's theorem.

We now use an approximation argument similar to the one used in the proof of Lemma~\ref{lem:parseval} to show that $h_y(x) = \widetilde{h}_y(x)$ for a.e. $x\in \R$ and any $0<y <1$, which completes the proof. Let $\widehat{\phi} \in C_c^\infty(\R)$ be a mollifier satisfying $\widehat{\phi}(0) = 1$ and set $\widehat{v}_\epsilon(x) = \widehat{v}(x)\widehat{\phi}(\epsilon k)$. Since $\widehat{v}_\epsilon$ has compact support, the function
\begin{align}
    h^\epsilon(z) \coloneqq \frac{1}{2\pi} \int_{\R} \widehat{v}_\epsilon(k) \mathrm{e}^{\ii zk} \mathrm{d}k
\end{align}
is analytic. From the two different splittings $\mathrm{e}^{\ii zk} = \mathrm{e}^{\ii xk} (\widehat{p}_y + q_y)(k) = \overline{\widehat{r}_z(k)} + s_z(k)$ and Plancherel's identity, it satisfies
\begin{align}
    h^\epsilon(z) = (p_y \ast v_\epsilon)(x) + \frac{1}{2\pi} \int_{\R} q_y(k) \widehat{v}_\epsilon \mathrm{e}^{\ii kx} \mathrm{d} x = \int_{\R} \overline{r_z(\xi)} v_\epsilon(\xi) \mathrm{d} \xi + \frac{1}{2\pi} \int_{\R} s_z(k) \widehat{v}_\epsilon(k) \mathrm{d} k.
\end{align}
Since $v_\epsilon = v \ast \phi_\epsilon \ra v$ strongly in $L^p(\R)$ (or weak-$\ast$ for $p=\infty$) as $\epsilon \downarrow 0$, we have
\begin{align*}
    p_y \ast v_\epsilon \ra p_y \ast v \quad \mbox{strongly in $L^p(\R)$ (or weak-$\ast$ for $p=\infty$)} \quad \mbox{and} \quad \inner{r_z, v_\epsilon} \ra \inner{r_z, v}
\end{align*}
for any $0<y<1$ and any $z \in S$, respectively. On the other hand, since $\widehat{v}_\epsilon \ra \widehat{v}$ in $L^2_2(\R)$, we have
\begin{align*}
    \widehat{v}_\epsilon q_y \ra \widehat{v} q_y \quad \mbox{in $L^2(\R)$}\quad \mbox{ and }\quad \inner{s_z, \widehat{v}_\epsilon} \ra \inner{s_z, \widehat{v}},
\end{align*}
again for any $0<y<1$ and $z \in S$. From these convergence results (and dominated convergence), we find that
\begin{align*}
    \inner{h_y,\psi} = \lim_{\epsilon \downarrow 0} \inner{h_y^\epsilon, \psi} = \inner{\widetilde{h}_y, \psi} \quad \mbox{for any Schwartz function $\psi \in \mathcal{S}(\R)$.}
\end{align*}
from which we infer that $h_y(x) = \widetilde{h}_y(x)$ for a.e. $x\in \R$.

The converse implication follows directly from Lemma~\ref{lem:existenceboundary}.
\end{proof}

\subsection{The variational problem over holomorphic functions}\label{sec:varpro-hol}

We can now use the previous lemma to reformulate our primal and dual problems over certain Hardy-like spaces.
In particular, we will present the proof of Theorem~\ref{thm:threelines}.

\begin{proposition}[Primal problem over holomorphic functions]\label{prop:holomorphicprimal} Let $M_\gamma$ be the value defined in \eqref{eq:primal} and let 
\begin{align}
	p_\gamma(z) \coloneqq \frac{1}{\pi} \frac{\frac{2}{\gamma}}{(\frac{2}{\gamma})^2 + z^2},  \label{eq:polefunctiondef} 
\end{align}
then the primal problem~\eqref{eq:primalR} is equivalent to
\begin{align}
	M_\gamma &=  \frac{4\pi}{\gamma} \inf \biggr\{ \norm{m_0}_{L^1(\R)}^{\gamma-2} \norm{m_1}_{L^2(\R)}^2  : m-p_\gamma \in \mathbb{H}^{1,2}(S)\biggr\}, \label{eq:primalholomorphic}
\end{align}
where $m_0(x) = m(x)$ and $m_1(x) = m(x+\ii)$. Moreover, the minimizer (provided it exists) is unique up to the transformation $m(z) \mapsto \mathrm{e}^{\alpha(\ii x+2/\gamma)} m(z)$ for $\alpha \in \R$.
\end{proposition}

\begin{remark}
\begin{enumerate}[label=(\arabic*)]
\item In other words, Proposition~\ref{prop:holomorphicprimal} says that the primal problem consists in finding the  best holomorphic approximation to the simple pole function $(z-\ii \frac{2}{\gamma})^{-1}$ with respect to suitable $L^p$ norms on the boundary of $S$.
\item In the formulation of Proposition~\ref{prop:holomorphicprimal}, the ansatz used for the primal problem in  \cite[Appendix D]{HKRV23} corresponds to
\begin{align*}
h(z)- p_\gamma(z) = \frac{1}{2\pi} \frac{\alpha^p \beta^q}{(\alpha-1-\ii z)^p(\beta-1-\ii z)^q(1+\ii z)}.
\end{align*}
\end{enumerate}
\end{remark}

\begin{proof}[Proof of Proposition~\ref{prop:holomorphicprimal}] 
Let $m\in\dom\mathcal{E}_{\gamma}$, i.e.\ $m\in L^1(\RR)$ such that $\widehat{m} - \exp \in L^2_{\gamma}$. By the change of variables $\widehat{\widetilde{m}}(k)\coloneqq \widehat{m}(\frac{2}{\gamma}k)$, which amounts to $\widetilde{m}(x) = \frac{\gamma}{2} m(\frac{\gamma}{2}x)$, we have that
\begin{align*}
	\dom\mathcal{E}_{\gamma} = \left\{ m \in L^1(\R): \widehat{\widetilde{m}} - \exp_{\frac{2}{\gamma}} \in L^2_2(\RR) \right\}.
\end{align*}

Next, note that $P_0(x) \coloneqq p_{\gamma}(x)$ for $x\in\RR$ satisfies $P_0 \in L^1(\RR)$ with $\widehat{P_0}(k)  = \mathrm{e}^{-\frac{2}{\gamma} |k|}$; in particular, since $\gamma > 2$,
\begin{align*}
	\int_{\RR} \left| \widehat{P_0}(k) - \mathrm{e}^{\frac{2}{\gamma} k}\right|^2 \,\mathrm{e}^{-2k}\,\mathrm{d}k 
	= \int_{0}^{\infty} \left| \mathrm{e}^{-\frac{2}{\gamma}k} - \mathrm{e}^{\frac{2}{\gamma}k} \right|^2 \,\mathrm{e}^{-2k}\,\mathrm{d}k < \infty.
\end{align*}
Hence, we may write 
\begin{align*}
	\dom \mathcal{E}_\gamma = \{m \in L^1(\R) : \widetilde{m} = P_0+v \mbox{ for } v \in L^1(\R), \widehat{v} \in L^2_2(\R)\}.
\end{align*}
Consequently, 
\begin{align}\label{eq:primalhere}
	M_\gamma &= \frac{2}{\gamma} \min_{\widetilde{m} \in L^1} \norm{\widetilde{m}}_{L^1}^{\gamma-2} \norm{\widehat{\widetilde{m}} - \exp_{\frac{2}{\gamma}}}_{L^2_2}^2 
	= \frac{2}{\gamma} \min_{\substack{v \in L^1 \\ \widehat{v} \in L^2_2}} \, \norm{P_0+v}_{L^1}^{\gamma-2} \bignorm{\widehat{P_0} +\widehat{v} -\exp_{\frac{2}{\gamma}}}_{L^2_2}^2.
\end{align}
We now consider the function $P_1(x)\coloneqq p_{\gamma}(x+\ii)$, which by the residue theorem satisfies 
\begin{align*}
	\widehat{P_1}(k) = \widehat{P_0}(k) \,\mathrm{e}^{-k} - 2\pi\ii \mathrm{Res}_{\ii \frac{2}{\gamma}}\left(p_{\gamma}\mathrm{e}^{-\ii k (z-\ii)}\right)
	= \widehat{P_0}(k) \,\mathrm{e}^{-k} - \mathrm{e}^{\frac{2}{\gamma}} \mathrm{e}^{-k}.
\end{align*}
Then by \eqref{eq:primalhere}, we can write 
\begin{align}
	M_{\gamma} 
	&= \frac{2}{\gamma} \min_{\substack{v \in L^1 \\ \widehat{v} \in L^2_2}} \norm{P_0+v}_{L^1}^{\gamma-2} \bignorm{\widehat{P_1}\exp + \exp_{\frac{2}{\gamma}} + \widehat{v} - \exp_{\frac{2}{\gamma}}}_{L^2_2}^2 \nonumber \\
	&=\frac{2}{\gamma} \min_{\substack{v \in L^1 \\ \widehat{v} \in L^2_2}} \norm{P_0+v}_{L^1}^{\gamma-2} \bignorm{\widehat{P_1} + \widehat{v} \exp_{-1}}_{L^2}^2. \label{eq:primalhere2}
\end{align}
Appealing to Theorem~\ref{thm:holomorphicextension}, $v$ is the boundary value $h_0$ of a unique holormophic function $h \in \mathbb{H}^{1,2}(S)$.
Since $\widehat{h_1} = \widehat{v} \exp_{-1}$ for such $h$, we obtain
\begin{align*}
		M_{\gamma} 	
		&=\frac{2}{\gamma} \min_{h \in \mathbb{H}^{1,2}(S)} \norm{P_0+h_0}_{L^1}^{\gamma-2} \bignorm{\widehat{P_1} + \widehat{h_1}}_{L^2}^2
		= \frac{4\pi }{\gamma} \min_{h \in \mathbb{H}^{1,2}(S)} \norm{P_0+h_0}_{L^1}^{\gamma-2} \bignorm{P_1 + h_1}_{L^2}^2,
\end{align*}
where in the last equality we used Plancherel's theorem. This completes the proof of \eqref{eq:primalholomorphic}.
\end{proof}

We can now give the proof of Theorem~\ref{thm:threelines}.
\begin{proof}[Proof of Theorem~\ref{thm:threelines}] 
Similar to the previous proof, note that for any $g \in L^\infty(\R)$ such that $\widehat{g} \in L^2_{-\gamma}(\RR)$ we have that the function $\widetilde{g}(x) \coloneqq \frac{\gamma}{2} g(-\frac{\gamma}{2} x)$ satisfies $\widetilde{g} \in L^\infty(\R)$ and $\widehat{\widetilde{g}} \in L^2_2(\RR)$. Moreover, there holds
\begin{align*}
	\norm{\widetilde{g}}_{L^\infty} 
	= \frac{\gamma}{2} \norm{g}_{L^\infty},
	\quad\norm{\widehat{\widetilde{g}}}_{L^2_2}^2 
	= \frac{\gamma}{2} \norm{\widehat{g}}_{L^2_{-\gamma}}^2,
	\quad\mbox{and}\quad 
	\re \inner{\widehat{\widetilde{g}}, \exp_{-\frac{2}{\gamma}} } 
	= \frac{\gamma}{2} \re \inner{\widehat{g}, \exp}.
\end{align*}
By Theorem~\ref{thm:holomorphicextension}, we can now find a unique $h \in \mathbb{H}^{\infty,2}(S)$ with $h_0 = \widetilde{g}$. The function $h$ satisfies $\widehat{h_1}(k) = \widehat{\widetilde{g}}(k) \mathrm{e}^{-k}$ for  $k\in\RR$, and therefore 
\begin{align*}
    \norm{\widehat{\widetilde{g}}}_{L^2_2}^2 = \|\widehat{h_1}\|_{L^2}^2 = 2\pi \|h_1\|_{L^2}^2.
\end{align*}
Moreover, for any $y\in (0,1)$ we have that $\widehat{h_y}(k) = \widehat{\widetilde{g}}(k) \,\mathrm{e}^{-y k}$ for $k\in\RR$, hence
\begin{align*}
	\re \inner{\widehat{\widetilde{g}}, \exp_{-\frac{2}{\gamma}} } 
	&= \re \int_{\RR} \widehat{\widetilde{g}} \,\mathrm{e}^{-\frac{2}{\gamma}k}\,\mathrm{d}k 
	=   \re \int_{\RR} \widehat{h_{\frac{2}{\gamma}}}(k)\,\mathrm{d}k 
	= 2\pi \re h_{\frac{2}{\gamma}}(0) 
	= 2\pi \re h(\ii \tfrac{2}{\gamma}).
\end{align*}
We can therefore recast \eqref{eq:scaleinvariantdual} into
\begin{align}\label{eq:holomorphicdual-proof}
	M_{\gamma} = \frac{16\pi (\gamma-2)^{\gamma-2}}{\gamma^{\gamma+1}} \left(\max_{h\in\mathbb{H}^{\infty,2}(S)} \frac{\Re h(\ii\tfrac{2}{\gamma})}{\|h_1\|_{L^2}^{\frac{2}{\gamma}} \|h_0\|_{L^{\infty}}^{1-\frac{2}{\gamma}}} \right)^{\gamma}
\end{align}
Up to a phase factor, $\re h(\ii\frac{2}{\gamma})$ can be replaced by $|h(\ii \frac{2}{\gamma})|$ in the maximization problem \eqref{eq:holomorphicdual-proof}, which by translation invariance of $\|h_1\|_{L^2}$ and $\|h_0\|_{L^{\infty}}$ can in turn be replaced by $\norm{h_{\frac{2}{\gamma}}}_{L^\infty}$. By passing from $h: S \to \CC$ to the function $\widetilde{h}: S_-\to \CC$, $z\mapsto \widetilde{h}(z)\coloneqq \overline{h(\overline{z})}$, we obtain \eqref{eq:3linesproblem}.
\end{proof}
\section{Characterization of optimizers} \label{sec:optimisers}

In this section we derive the Euler-Lagrange equation for the primal problem and use it to characterize the primal and dual optimizers. One of the key observations is that the primal and dual optimizers can be seen as analytic/meromorphic extensions of each other.

\subsection{Euler-Lagrange equation and dual-primal optimizer relation}
Let us start with the Euler-Lagrange equation for the functional $\mathcal{E}_{\gamma}$ associated with the primal problem~\eqref{eq:primalL1},
\begin{align*}
    \mathcal{E}_\gamma(m) = \norm{m}_{L^1}^{\gamma-2} \norm{\widehat{m}-\exp}_{L^2_\gamma}^2.
\end{align*}
Formally, the Gateaux derivative of $\mathcal{E}_\gamma$ is given by
\begin{align}
	\mathrm{d}_m \mathcal{E}_\gamma(\delta) 
	= (\gamma-2) \|m\|_{L^1}^{\gamma-3} \|\widehat{m}-\exp\|_{L^2_{\gamma}}^2 \Re \inner{\textrm{sign}\, m,\delta} + 2 \|m\|_{L^1}^{\gamma-2} \re \biggr\langle \frac{\widehat{m} - \exp}{\exp_{\gamma}}, \widehat{\delta}\biggr\rangle,  \label{eq:Euler-Lagrange}
\end{align} 
where $\mathrm{sign}\, z = \frac{z}{|z|}$, $z\in \CC\setminus\{0\}$, is the complex sign function. 
To make this calculation rigorous, however, the critical step is to show that $m(x) \neq 0$ almost everywhere. It turns out that, by using finer results in the theory of Hardy spaces, one can show that this holds for any $m \in \dom \mathcal{E}_\gamma$. This calculation, however, is rather technical and not needed for our purposes in this article. Instead, we show here that $m$ is a minimizer of $\mathcal{E}_\gamma$, provided it satisfies the following Euler-Lagrange equation. Recall from equation \eqref{eq:Edomain} that $\dom \mathcal{E}_\gamma = \{m\in L^1(\RR): \widehat{m} -\exp \in L^2_{\gamma}\}$.
\begin{lemma}[Euler-Lagrange equation] \label{lem:eulerlagrange} Let $m \in \dom \mathcal{E}_\gamma$ satisfy $m(x) \neq 0$ a.e.\ and the Euler-Lagrange equation
\begin{align}
	\widehat{\mathrm{sign}\, m}(k) + c \left(\widehat{m}(k) - \mathrm{e}^k\right)\mathrm{e}^{-\gamma k} = 0 \quad \text{for a.e. } k \in \RR, \label{eq:primaloptimizer}
\end{align}
for some constant $c>0$, then $m$ is a minimizer of the primal problem~\eqref{eq:primalL1}. Moreover, for any $\beta >0$, the function
\begin{align}
  g \coloneqq \beta \, \mathrm{sign}\, m,  \label{eq:primaldualrelation}
\end{align}
is a maximizer of \eqref{eq:scaleinvariantdual}.
\end{lemma}

\begin{proof} First note that since $m \neq 0$ a.e., we have $\mathrm{sign} \,m \in L^\infty(\R)$; in particular, the Fourier transform $\widehat{\mathrm{sign}\, m}$ is well-defined in the tempered distribution sense. Since $m \in \dom\mathcal{E}_{\gamma}$,  equation~\eqref{eq:primaloptimizer} therefore makes sense as a distributional identity and implies that $\widehat{\mathrm{sign} \, m} \in L^2_{-\gamma}(\R)$, from which we infer that \eqref{eq:primaloptimizer} also holds pointwise for a.e.\ $k\in\RR$. Hence, if we set
\begin{align*}
\mathcal{E}^c_\gamma(m) \coloneqq \norm{m}_{L^1} + \frac{c}{4 \pi} \norm{\widehat{m}-\exp}_{L^2_\gamma}^2
\end{align*}
where $c>0$ is the constant in \eqref{eq:primaloptimizer}, from Plancherel's identity (Lemma~\ref{lem:parseval}) we conclude that
\begin{align*}
    \lim_{\epsilon \downarrow 0} \frac{\mathcal{E}_\gamma^c(m+\epsilon \delta) - \mathcal{E}_\gamma^c(m)}{\epsilon} &= \re \inner{ \mathrm{sign}\, m, \delta} + \frac{c}{2\pi} \re \inner{(\widehat{m} - \exp)\exp_{-\gamma},\widehat{\delta}} \\
    &= \frac{1}{2\pi} \inner{ \widehat{\mathrm{sign} \, m} + c (\widehat{m} - \exp)\exp_{-\gamma}, \widehat{\delta}} = 0.
\end{align*}
In particular, $m$ is a critical point of $\mathcal{E}_\gamma^c$. Consequently, $m$ must be the global minimizer of $\mathcal{E}^c_\gamma$ by the strict convexity of $\mathcal{E}^c_\gamma$, and therefore, a minimizer of $\mathcal{E}_{\gamma}$ by Lemma~\ref{lem:scaling}.

For the statement about the dual optimizer we set $g \coloneqq \beta \mathrm{sign} \, m$ for some $\beta >0$. Then, since $m$ is a primal optimizer and satisfies $m \neq 0$, we have
\begin{align*}
    \mathrm{d}_m  \mathcal{E}_\gamma(\delta) = (\gamma-2) \|m\|_{L^1}^{\gamma-3} \|\widehat{m}-\exp\|_{L^2_{\gamma}}^2 \Re \inner{\textrm{sign}\, m,\delta} + 2 \|m\|_{L^1}^{\gamma-2} \re \biggr\langle \frac{\widehat{m} - \exp}{\exp_{\gamma}}, \widehat{\delta}\biggr\rangle = 0.
\end{align*}
From this equation, Plancherel's identity in the version of Lemma~\ref{lem:parseval}, and \eqref{eq:primaloptimizer}, we conclude that the constant $c>0$ in \eqref{eq:primaloptimizer} is given by
\begin{align}
    c = \frac{4 \pi}{\gamma-2} \frac{\norm{m}_{L^1}}{\norm{\widehat{m} - \exp}_{L^2_\gamma}^2} >0. \label{eq:mconst}
\end{align}
In particular, we have
\begin{align*}
\norm{\widehat{g}}_{L^2_{-\gamma}}^2 &= \beta^2 c^2 \norm{\widehat{m} - \exp}_{L^2_\gamma}^2 = \beta^2 \frac{(4\pi)^2}{(\gamma-2)^2} \frac{\norm{m}_{L^1}^2}{\norm{\widehat{m}-\exp}_{L^2_\gamma}^2}.
\end{align*}
Moreover, there holds 
\begin{align*}
	\beta^{-1} \inner{\widehat{g} , \exp- \widehat{m}} 
	&= - c\inner{(\widehat{m} - \exp)\exp_{-\gamma}, \exp-\widehat{m}} 
	= c \norm{\widehat{m}-\exp}_{L^2_\gamma}^2 
	= \frac{4\pi}{\gamma-2} \norm{m}_{L^1},
\end{align*}
hence
\begin{align*}
	\inner{\widehat{g} , \exp} 
	&= \beta \frac{4\pi}{\gamma-2} \norm{m}_{L^1} + \inner{\widehat{g} , \widehat{m}} 
	= \beta \frac{4\pi}{\gamma-2} \norm{m}_{L^1} + 2\pi \inner{g,m} 
	= \beta \left(\frac{4\pi}{\gamma-2} \norm{m}_{L^1} + 2\pi \|m\|_{L^1}\right)\\
	&= \beta \frac{2\pi \gamma}{\gamma-2} \|m\|_{L^1}.
\end{align*}
So, using that $\norm{g}_{L^\infty} = \beta$, and substituting the above values in the dual problem \eqref{eq:scaleinvariantdual} we find
\begin{align*}
\mathcal{E}^\ast_\gamma(g) 
	= \frac{4 (\gamma-2)^{\gamma-2}}{(2\pi)^{\gamma-2} \gamma^\gamma} \frac{\re \inner{\widehat{g} , \exp}^\gamma}{\norm{\widehat{g}}_{L^2_{-\gamma}}^2 \norm{g}_{L^\infty}^{\gamma-2}} 
	= \norm{m}_{L^1}^{\gamma-2} \norm{\widehat{m}-\exp}_{L^2_\gamma}^2 
	= \mathcal{E}_{\gamma}(m)
	= M_\gamma,
\end{align*}
which implies that $g$ is an optimizer by the strong duality in Theorem~\ref{thm:scaleinvariantdual}.
\end{proof}

In the holomorphic formulation, the above Euler-Lagrange equations can be stated as follows.

\begin{lemma}[Euler-Lagrange - Holomorphic version] \label{lem:holomorphicEulerLagrange} Suppose that $m \in p_\gamma + \mathbb{H}^{1,2}(S)$ can be meromorphically extended to $S_2 \coloneqq \{x + \ii y : 0<y<2\}$ and satisfies the following properties:
\begin{enumerate}[label=(\roman*)]
\item\label{it:mbound} $m$ can be continuously extended to the boundary $\partial S_2$, and the function $m-p_\gamma$ is bounded on $\overline{S_2}$,
\item $m$ satisfies the Euler-Lagrange equation
\begin{align}
	\mathrm{sign}\, m_0 = -c m_2 ,\quad \mbox{for some $c>0$.}  \label{eq:ELholomorphic}
\end{align}
\end{enumerate}
Then $m$ is a minimizer of \eqref{eq:primalholomorphic}. Moreover, for any $\beta \in \C \setminus \{0\}$, the function 
\begin{align}
	h(z) \coloneqq \beta m(z+2\ii) \quad\mbox{ for all $z \in S_{-2}$,} \label{eq:holomorphicprimaldual}
\end{align}
is a maximizer of \eqref{eq:3linesproblem}.
\end{lemma}

\begin{proof} Let $\phi \in C_c^\infty(\R)$, and notice that $\widehat{\phi}$ is an entire function by the Paley-Wiener theorem. By assumption, $m$ has a meromorphic extension to $S_2$, which we also denote by $m$, and this extension is holomorphic in $S_{(1,2)} = \{x+\ii y \in \C: 1<y<2\}$. In particular, the function $z \mapsto \overline{m(\overline{z}+2\ii)} \widehat{\phi}(z)$ is holomorphic in the strip $S$ and continuous up to the boundary. Moreover, this function is bounded by assumption~\ref{it:mbound} and has integrable decay as $|x| \ra \infty$. Consequently, we can apply the Cauchy integral theorem (c.f.\ \eqref{eq:cauchytheorem}) and use \eqref{eq:ELholomorphic} to obtain
\begin{align}\label{eq:el-1}
\inner{\mathrm{sign}\, m_0, \widehat{\phi}_0}  = -c \inner{m_2 , \widehat{\phi}_0} = -c \inner{m_1, \widehat{\phi}_1} = -c\inner{(m-p_\gamma)_1,\widehat{\phi}_1} -c \inner{(p_\gamma)_1, \widehat{\phi}_1}
\end{align}
for any $\phi \in C_c^\infty(\R)$. Note that since $m-p_{\gamma} \in \mathbb{H}^{1,2}(S)$ and by assumption~\ref{it:mbound}, the function $z \mapsto \overline{(m-p_\gamma)(\overline{z}+\ii)} \widehat{\phi}(z)$ is holomorphic in $S$, bounded and continuous on $\overline{S}$, and has integrable decay as $|x| \ra \infty$. Hence, by \eqref{eq:cauchytheorem} applied to the function $m-p_{\gamma}$, there holds
\begin{align}\label{eq:el-2}
    \inner{(m-p_\gamma)_1,\widehat{\phi}_1} =  \inner{(m-p_\gamma)_0,\widehat{\phi}_2} = \inner{m_0, \widehat{\phi}_2} - \inner{(p_\gamma)_0,\widehat{\phi}_2}.
\end{align}
Recalling the definition of $p_\gamma$ in \eqref{eq:polefunctiondef}, we see that the function $f_{\gamma}: z \mapsto \overline{(p_\gamma)(\overline{z}+\ii)} \widehat{\phi}(z)$ is holomorphic on $\CC \setminus\{z_1, z_2\}$ with two isolated poles of first order at $z_{1,2} = \ii (1 \mp \frac{2}{\gamma})$. Hence, the residue theorem implies that 
\begin{align}\label{eq:el-3}
    \inner{(p_\gamma)_1, \widehat{\phi}_1} = \inner{(p_\gamma)_0, \widehat{\phi}_2} + 2\pi \ii\, \mathrm{Res}_{z_1} f_{\gamma} 
    = \inner{(p_\gamma)_0, \widehat{\phi}_2} - \widehat{\varphi}_{2-\frac{2}{\gamma}}(0).
\end{align}
Combining \eqref{eq:el-1}, \eqref{eq:el-2}, and \eqref{eq:el-3}, it follows that
\begin{align}\label{eq:almostel}
    \inner{\mathrm{sign}\, m_0, \widehat{\phi}_0} = - c \inner{m_0, \widehat{\phi}_2} + c \widehat{\varphi}_{2-\frac{2}{\gamma}}(0).
\end{align}

Now note that from~\eqref{eq:ELholomorphic}, we have $\mathrm{sign} \, m_0 \in L^\infty(\R)$, in particular, $\mathrm{sign} \, m_0$ is a tempered distribution and its Fourier transform is well-defined. On the other hand, $m_0 \in L^1(\R)$, so that by the standard Plancherel identity, we can rewrite \eqref{eq:almostel} as
\begin{align*} 
     \inner{\mathrm{sign} \, m_0, \widehat{\phi}_0} = -c \int_\R \left(2\pi\overline{ \widecheck{m}_0(x)} - \mathrm{e}^{-\frac{2}{\gamma}x}\right) \mathrm{e}^{2x} \phi(x) \mathrm{d} x, \quad \mbox{for any $\phi \in C_c^\infty(\R).$}
\end{align*}
Since $C_c^\infty(\R)$ is dense in $\mathcal{S}(\R)$, the (inverse) Fourier transform of $\mathrm{sign} \, m_0$ is uniquely defined by the above equation and satisfies
\begin{align*}
2\pi \widecheck{\mathrm{sign}\, m_0}(x) = - c \left(2\pi \widecheck{m_0}(x) - \mathrm{e}^{-\frac{2}{\gamma}x}\right) \mathrm{e}^{2x} \quad \mbox{for a.e. $x \in \R$.} 
\end{align*}
Therefore, using that $\widecheck{m}_0(x) = \frac{1}{2\pi} \widehat{m}_0(-x)$ for $x\in\RR$, the rescaled function $\widetilde{m}_0(x) = \tfrac{2}{\gamma} m_0(\tfrac{2}{\gamma} x)$ satisfies the Euler--Lagrange  equation~\eqref{eq:primaloptimizer} from Lemma~\ref{lem:eulerlagrange}. Thus, $\widetilde{m}$ is a minimizer of the primal problem~\eqref{eq:primalL1} and from its correspondence with the holomorphic version~\eqref{eq:primalholomorphic}, we conclude that $m$ is a minimizer of the latter.

We now prove the statement on the dual optimizer: First, notice that the dual problem~\eqref{eq:3linesproblem} is invariant under a phase change, in particuar, we can assume $\beta>0$ in \eqref{eq:holomorphicprimaldual} without loss of generality. Next, notice that the re-scaled boundary value $g(x) \coloneqq \tfrac{2}{\gamma} h_0(\tfrac{2}{\gamma} x)$ satisfies
$g(x)= -c \beta \,\mathrm{sign} \, m_0(\tfrac{2}{\gamma} x)$ by \eqref{eq:ELholomorphic}. As $\tfrac{2}{\gamma} m_0(\tfrac{2}{\gamma} x)$ satisfies~\eqref{eq:primaloptimizer}, the function $g$ satifies \eqref{eq:primaldualrelation}, and the proof follows from the statement for the dual optimizer in Lemma~\ref{lem:eulerlagrange}.
\end{proof}

\subsection{Optimizers}

We now give an explicit construction of the optimizers of our variational problem. To this end, we use the following lemma to construct a phase function $\theta$ (on $\R$) whose analytic extension satisfies a logarithmic variant of the Euler-Lagrange equation~\eqref{eq:ELholomorphic}. 
\begin{definition} \label{def:pvdistribution}
	Let $g \in L^1(\RR)$ be an even function. 
	For any $\varphi\in\mathcal{S}(\RR)$, we define
	\begin{align}
		\eta_g(\varphi) &\coloneqq \int_0^\infty \frac{g(k)}{k (\cosh(2k)-1)} (\varphi(k) - \varphi(-k) - 2\varphi'(0)k)\,\mathrm{d}k,
	\intertext{and}
		\widetilde{\eta}_g(\varphi) &\coloneqq \int_0^\infty \frac{g(k)}{k} (\varphi(k) - \varphi(-k))\,\mathrm{d}k,
	\end{align}
\end{definition}

\begin{lemma}\label{lem:principal-value-distributions}
	Let $g\in L^1(\RR)$ be an even function. Then $\eta_g, \widetilde{\eta}_g \in \mathcal{S}'(\RR)$ define tempered distributions. 
	Moreover, for any $|y|\leq 2$, 
	\begin{align*}
		(\eta_g \exp_{-y})(\varphi) \coloneqq \eta_g(\exp_{-y}\varphi), \quad \varphi \in \mathcal{S}(\RR),
	\end{align*}
	is a well-defined tempered distribution. 
\end{lemma}

\begin{remark}
	The tempered distributions $\eta_g$ and $\widetilde{\eta}_g$ can also be defined via the principal value integrals 
	\begin{align*}
		\eta_g(\varphi) &= \lim_{\epsilon\downarrow0} \int_{|k|\geq \epsilon} \frac{g(k)}{k (\cosh(2k)-1)} (\varphi(k)-\varphi'(0)k)\,\mathrm{d}k \\
		&=  \lim_{\epsilon\downarrow0} \int_{\epsilon}^{\infty} \frac{g(k)}{k (\cosh(2k)-1)} (\varphi(k)-\varphi(-k)-2\varphi'(0)k)\,\mathrm{d}k,
	\end{align*}
	and 
	\begin{align*}
		\widetilde{\eta}_g(\varphi) &= \lim_{\epsilon\downarrow0} \int_{|k|\geq \epsilon} \frac{g(k)}{k} \varphi(k)\,\mathrm{d}k 
		= \lim_{\epsilon\downarrow0} \int_{\epsilon}^{\infty} \frac{g(k)}{k} (\varphi(k)-\varphi(-k))\,\mathrm{d}k,
	\end{align*}
	for any Schwartz function $\varphi\in\mathcal{S}(\RR)$. Hence $\widetilde{\eta}_g$ is the Cauchy principal value distribution associated to $k\mapsto\frac{g(k)}{k}$.
\end{remark}

\begin{proof}[Proof of Lemma~\ref{lem:principal-value-distributions}]
	Let $g\in L^1(\RR)$ and $|y| \leq 2$. Then for any $\varphi\in\mathcal{S}(\RR)$ we can estimate
	\begin{align}
		\biggr| \int_{1}^{\infty} \frac{g(k)}{k (\cosh(2k)-1)} &\bigr(\,\mathrm{e}^{-y k}\varphi(k) +  \mathrm{e}^{y k} \varphi(-k) +2y \varphi(0) k - 2\varphi'(0)k\bigr)\,\mathrm{d}k \biggr| \nonumber \\
		&\lesssim \sup_{k\geq 1} \biggr| \frac{\mathrm{e}^{-y k} + \mathrm{e}^{y k} + 2 y k + 2 k}{ k(\cosh(2k)-1)} \biggr| \bigr(\norm{\varphi}_{L^\infty(\R)} + \norm{\varphi'}_{L^\infty(\R)} |\bigr)\norm{g}_{L^1(\R)}.  \nonumber \\
  &\lesssim \bigr(\norm{\varphi}_{L^\infty(\R)} + \norm{\varphi'}_{L^\infty(\R)}\bigr)\norm{g}_{L^1(\R)}. \label{eq:eta-1}
	\end{align}
	Moreover, setting $\psi = \exp_{-y} \varphi$ and noting that the map $k \mapsto \psi(k) - \psi(-k) -2 \psi'(0)k$ is odd, we have 
	\begin{align*}
		\left|\psi(k) - \psi(-k) - 2 \psi'(0) k \right|
		= \left| \int_{-k}^k \int_0^s \int_0^t \psi'''(r)\,\mathrm{d}r\,\mathrm{d}t\,\mathrm{d}s \right| 
		\leq \|\psi'''\|_{L^{\infty}[-1,1]} \frac{k^3}{3}
	\end{align*}
	for any $k \in [0,1]$, hence
	\begin{align}\label{eq:eta-2}
		&\left| \int_0^1 \frac{g(k)}{k (\cosh(2k)-1)} (\psi(k) - \psi(-k) - 2\psi	'(0)k)\,\mathrm{d}k \right| \\
		&\quad \lesssim \sup_{k\in(0,1)} \left| \frac{k^2}{\cosh(2k)-1} \right| \,  \|\psi'''\|_{L^{\infty}[-1,1]} \|g\|_{L^1(\RR)} 
		\lesssim  \|\psi'''\|_{L^{\infty}[-1,1]} \|g\|_{L^1(\RR)}.
	\end{align}
	By the definition of $\psi$, it follows from the Leibniz rule and $|y|\leq 2$ that 
	\begin{align*}
		 \|\psi'''\|_{L^{\infty}[-1,1]} 
		 \leq \sup_{k\in[-1,1]} \mathrm{e}^{-y k} \left| \varphi'''(k) - 3 y \varphi''(k) + 3 y^2 \varphi'(k) - y^3 \varphi(k) \right| 
		 \lesssim \sum_{j=0}^3 \|\varphi^{(j)}\|_{L^{\infty}(\RR)}.
	\end{align*}
	Together with \eqref{eq:eta-2} and \eqref{eq:eta-1}, this shows that $\eta_g \exp_{-y}$ is a tempered distribution for any $|y|\leq 2$, in particular that $\eta_g$ defines a tempered distribution. The easy proof for $\widetilde{\eta}_g$ is standard.
\end{proof}

The phase function $\theta$ can now be defined via inverse Fourier transform of the tempered distribution $\eta_g$, which can be extended to the complex plane owing to the fact that $\eta_g \exp_{-y}$ is a tempered distribution for any $|y|\leq 2$.

\begin{lemma}[Singular principal value phase function]\label{lem:singularPVphasefunction} 
Let $g\in L^1(\RR)$ be an even function with bounded derivatives up to order $3$ on $[-1,1]$. Then the function $z\mapsto \theta_{g}(z):= \frac{1}{2\pi} \eta_g(\mathrm{e}^{\ii z (\cdot)})$ is well-defined and holomorphic on the open strip $S_{(-2,2)} = \{ x+\ii y \in \C: |y| <2 \}$. 
Moreover, it has a continuous extension to the closure of $S_{(-2,2)}$ satisfying
\begin{align}
	\ii\mathrm{Im}\, \theta_{g}(x-2\ii) - \theta_{g}(x) = \frac{1}{2\pi} \widetilde{\eta}_g(\mathrm{e}^{\ii x (\cdot)}), \label{eq:phasefunctionid}
\end{align}
and the estimate
\begin{align}
 &\left|\re \theta_{g}(x+\ii y) - g(0) \frac{|x| y}{2}\right| \lesssim \|g\|_{L^1} +  \sum_{j=0}^3 \|g^{(j)}\|_{L^{\infty}([0,1])}, \label{eq:realthetaest}
\end{align}
 for any $x+\ii y \in \overline{S_{(-2,2)}}$.
\end{lemma}

\begin{proof} 

For the estimate \eqref{eq:realthetaest}, note that for $x+ \ii y \in \overline{S_{(-2,2)}}$ we have
\begin{align*}
	\Re \theta_{g}(x+\ii y) &= -\frac{1}{\pi} \int_1^\infty \frac{g(k) (\cos(kx)\sinh(ky)-ky)}{k(\cosh(2k)-1)} \,\mathrm{d}k \\
	&\quad- \frac{1}{\pi} \int_0^1 \frac{g(k) (\cos(kx)\sinh(ky)-ky)}{k(\cosh(2k)-1)} \,\mathrm{d}k.
\end{align*}
Clearly, since $|y| \leq 2$, the first term is bounded by
\begin{align*}
	\frac{1}{\pi} \left| \int_1^\infty \frac{g(k) (\cos(kx)\sinh(ky)-ky)}{k(\cosh(2k)-1)} \,\mathrm{d}k \right| 
	\lesssim \|g\|_{L^1}.
\end{align*}
For the second term, we write 
\begin{align}
	&\frac{1}{\pi} \int_0^1 \frac{g(k)(\cos(kx)\sinh(ky)-ky)}{k(\cosh(2k)-1)} \,\mathrm{d}k \nonumber \\
	&=\frac{1}{\pi} \int_0^1 \frac{g(k) (\cos(kx)-1) \sinh(ky)}{k(\cosh(2k)-1)} \,\mathrm{d}k \nonumber 
	+ \frac{1}{\pi} \int_0^1 \frac{g(k) (\sinh(ky)-ky)}{k(\cosh(2k)-1)}\,\mathrm{d}k \nonumber,
\end{align}
and use that since
\begin{align*}
	\sup_{k\in(0,1)}\left| \frac{\sinh(k y)-ky}{k (\cosh(2k)-1)} \right| \leq \frac{|y|^3}{12},
\end{align*}
the latter expression is bounded by 
\begin{align*}
	\frac{1}{\pi}\left| \int_0^1 \frac{g(k)(\cos(kx)-1) \sinh(ky)}{k(\cosh(2k)-1)}  \,\mathrm{d}k \right| \lesssim \|g\|_{L^1}.
\end{align*}

To estimate the remaining term, we take an even, smooth cut-off function $\chi\in C_c^{\infty}(-1,1)$ with $|\chi| \leq 1$ and $\chi\equiv 1$ on $[-\frac{1}{2}, \frac{1}{2}]$, and split
\begin{align}
	&\frac{1}{\pi} \int_0^1 \frac{g(k) (\cos(kx)-1) \sinh(ky)}{k(\cosh(2k)-1)} \,\mathrm{d}k \nonumber \\
	&= \frac{1}{\pi} \int_0^1 \chi(k) \frac{g(k) \sinh(ky)}{k(\cosh(2k)-1)} (\cos(kx)-1) \,\mathrm{d}k \label{eq:splitting-1}\\
	&\quad + \frac{1}{\pi} \int_0^1 (1-\chi(k)) \frac{g(k) \sinh(ky)}{k(\cosh(2k)-1)} (\cos(kx)-1)\,\mathrm{d}k \label{eq:splitting-2}.
\end{align}
The term \eqref{eq:splitting-2} is again bounded by $|\eqref{eq:splitting-2}| \lesssim \|g\|_{L^{1}}$ (uniformly in $x$), while for \eqref{eq:splitting-1} we have to extract the behaviour in $x$ more carefully. To this end, we write 
\begin{align*}
	\cos(kx) - 1 = - \int_0^1 k x \sin(k x t)\,\mathrm{d}t = -\int_0^1 \int_0^1 k^2 x^2 t \cos(k x t s)\,\mathrm{d}s\,\mathrm{d}t 
\end{align*}
and obtain
\begin{align}
	& \int_0^1 \chi(k) \frac{g(k) \sinh(ky)}{k(\cosh(2k)-1)} (\cos(kx)-1)\,\mathrm{d}k \nonumber \\
	&= - x^2  \int_0^1 \int_0^1 \int_0^1 \chi(k) g(k)\frac{k \sinh(ky)}{\cosh(2k)-1} t \cos(kxts) \,\mathrm{d}s\,\mathrm{d}t \,\mathrm{d}k. \label{eq:splitting-3}
\end{align}
Note that the function $\psi_y: \RR \to \RR$, $\psi_y(k) \coloneqq \chi(k) g(k)\frac{k \sinh(ky)}{\cosh(2k)-1}$ for $k\in\RR$, is even, so that interchanging the order of integration (which is justified by the absolute integrability of the integrand), we may write 
\begin{align*}
	\eqref{eq:splitting-3} 
	&= - x^2 \int_0^1 \int_0^1 \widehat{\psi_y}(|x| t s) t \,\mathrm{d}s\,\mathrm{d}t 
	= - |x| \int_0^1 \int_0^{|x| t} \widehat{\psi_y}(s) \,\mathrm{d}s\,\mathrm{d}t \\
	&= -|x| \int_0^{\infty} \widehat{\psi_y}(s) \,\mathrm{d}s + |x| \int_0^1 \int_{|x| t}^{\infty} \widehat{\psi_y}(s) \,\mathrm{d}s \,\mathrm{d}t
	= -|x| \pi \psi_y(0) + |x| \int_0^1 \int_{|x| t}^{\infty} \widehat{\psi_y}(s) \,\mathrm{d}s \mathrm{d}t.
\end{align*}
At this place the regularity of $g$ comes into play: first, in order to evaluate
\begin{align*}
	\psi_y(0) = \chi(0) g(0) \lim_{k\to 0} \frac{k \sinh(ky)}{\cosh(2k)-1} 
	= g(0) \frac{y}{2},
\end{align*}
and secondly, to extract decay of $\widehat{\psi_y}$. Indeed, $\psi_y$ is a compactly supported $C^2$ function on $(-1,1)$ with bounded third derivative, hence 
\begin{align*}
	|\widehat{\psi_y}(s)| \lesssim \sum_{j=0}^3 \|\psi_y^{(j)}\|_{L^1} (1+s)^{-3}
\end{align*} 
for any $s\geq 0$. Using that $k\mapsto \frac{k \sinh(ky)}{(\cosh(2k)-1)}$ is uniformly bounded on $\RR$ for $|y|\leq 2$, it follows that 
\begin{align*}
	|\widehat{\psi_y}(s)| \lesssim \sum_{j=0}^3 \|g^{(j)}\|_{L^{\infty}([0,1])} (1+s)^{-3}, 
\end{align*}
in particular
\begin{align*}
	|x| \int_0^1 \int_{|x| t}^{\infty} \widehat{\psi_y}(s) \,\mathrm{d}s \mathrm{d}t 
	\lesssim \left( \sum_{j=0}^3 \|g^{(j)}\|_{L^{\infty}([0,1])} \right) |x| \int_0^1 \int_{|x| t}^{\infty} (1+s)^{-3} \,\mathrm{d}s \,\mathrm{d}t
	\lesssim \sum_{j=0}^3 \|g^{(j)}\|_{L^{\infty}([0,1])} .
\end{align*}

Hence, we have shown that 
\begin{align*}
	\left| \Re \theta_g(x+\ii y) - g(0) \frac{|x| y}{2} \right| \lesssim \|g\|_{L^1} +  \sum_{j=0}^3 \|g^{(j)}\|_{L^{\infty}([0,1])}. 
\end{align*}

In order to prove the relation \eqref{eq:phasefunctionid}, note that 
\begin{align*}
	\mathrm{Im}\, \theta_{g}(x-\ii y) 
	&= \frac{1}{\pi} \int_0^\infty \frac{g(k)(\sin(kx)\cosh(ky)-kx)}{k(\cosh(2k)-1)} \,\mathrm{d}k,
\end{align*} 
and
\begin{align*}
	\theta_g(x) 
	= \frac{\ii}{\pi} \int_0^\infty \frac{g(k) (\sin(kx)-kx)}{k(\cosh(2k)-1)}\,\mathrm{d}k.
\end{align*}
It follows that
\begin{align*}
	\ii \mathrm{Im}\,\theta_g(x-2\ii) - \theta_g(x) 
	&=  \frac{\ii}{\pi} \int_0^\infty \frac{g(k)}{k} \,\sin(kx)\,\mathrm{d}k
	= \frac{1}{2\pi} \widetilde{\eta}_g(\mathrm{e}^{\ii x (\cdot)}).
\end{align*}
\end{proof}

We are now in position to construct the optimizers of our variational problem.

\begin{theorem}[Analytic formula for optimizers]\label{thm:dualoptimizer} Up to the transformation $ h_{\alpha,\beta,\kappa}(z) = \beta h(z-\kappa) \mathrm{e}^{\ii \alpha z}$, $\alpha \in \R, \beta \in \C\setminus \{0\}, \kappa \in \R$, 
the unique optimizer $h\in \mathbb{H}^{\infty,2}(S_-)$ of \eqref{eq:3linesproblem} is given by the formula
\begin{align}
	h(z) = B_\gamma(z) \mathrm{e}^{\theta_\gamma(z)}, \quad \mbox{where } \quad B_\gamma(z) = \frac{z-\ii (2-\frac{2}{\gamma})}{z+\ii (2-\frac{2}{\gamma})} \label{eq:dualoptimizer}
\end{align}
is the Blaschke factor (on the upper half plane) with zero at $\ii(2-\frac{2}{\gamma})$, and $\theta_{\gamma}\coloneqq \frac{1}{2\pi} \eta_{g_\gamma}(\mathrm{e}^{\ii z(\cdot)})$ with $\eta_g$ given by Definition~\ref{def:pvdistribution}, and
\begin{align}
	g_\gamma(k) \coloneqq \pi \left(2 \mathrm{e}^{-(2-\frac{2}{\gamma})|k|} + \mathrm{e}^{-\frac{2}{\gamma}|k|} - \mathrm{e}^{-(4-\frac{2}{\gamma})|k|}\right). \label{eq:principalvaluefunction}
\end{align}
\end{theorem}

\begin{remark}\label{rem:properties-g}
	A simple calculation shows that for $\gamma>2$ the function $g_{\gamma}$ is two times continuously differentiable on $\RR$ with $g_{\gamma}(0) = 2\pi$. The third derivative of $g_{\gamma}$ is continuous away from the origin but has a jump discontinuity at zeros with $g_{\gamma}'''(0+) = 24\pi (2-\frac{2}{\gamma}) = -g_{\gamma}'''(0-)$. By the exponential decay in both directions it follows that $g_{\gamma}^{(j)} \in L^1(\RR) \cap L^{\infty}(\RR)$ for $j\leq 3$. 
\end{remark}

\begin{remark}\label{rem:blaschke}
	Note that for $\alpha>0$, the Blaschke factor $b_{\alpha}(z) \coloneqq \frac{z-\ii \alpha}{z+\ii\alpha}$ has the following properties: 
	\begin{enumerate}[label=(\roman*)]
		\item $b_{\alpha}$ has a simple pole at $-\ii \alpha$ in the lower half plane and a simple zero at $\ii \alpha$ in the upper half plane.
		\item $|b_{\alpha}(x)| = 1$ for $x\in\RR$, while $|b_{\alpha}(z)| \leq 1$ for $z\in \CC_+$; in particular, 
		\begin{align*}
			\sgn b_{\alpha}(x) = b_{\alpha}(x) = \exp\left( -2\ii \arctan \frac{\alpha}{x} \right) \quad \text{for } x\in\RR.
		\end{align*}
		\item For $x,y \in \RR$ such that $x+\ii y \neq \ii \alpha$, there holds
		\begin{align*}
			\sgn b_{\alpha}(x+\ii y) = \begin{cases}
				\exp\left[ \ii \left( \arctan\frac{y-\alpha}{x} - \arctan\frac{y+\alpha}{x}\right)\right], & x\neq 0,\\
				\sgn \frac{y-\alpha}{y+\alpha}, & x=0.
			\end{cases}
		\end{align*}
	\end{enumerate}
\end{remark}

\begin{proof} 
In view of Remark~\ref{rem:properties-g} and Lemma~\ref{lem:singularPVphasefunction}, the function $\theta_{\gamma}$ is holomorphic on $S_{(-2,2)}$ with a continuous extension to the boundary of the strip. Since the Blaschke factor $B_{\gamma}$ is meromorphic on $\CC$ with a single simple pole at $-\ii(2-\frac{2}{\gamma}) \in S_{(-2,2)}$, the function $h$ is meromorphic on  $S_{(-2,2)}$. 
Moreover, by \eqref{eq:realthetaest}, 
\begin{align}
	|\mathrm{e}^{\theta_{\gamma}(x+\ii y)}| = \mathrm{e}^{\re \theta_{\gamma}(x+\ii y)} \lesssim \mathrm{e}^{-\pi |x| |y|} \quad \text{for } y\in(-2,0). \label{eq:expdecay}
\end{align}
Consequently, 
\begin{align}
    h_y \in L^1 \cap L^\infty \quad \mbox{ for any $-2<y<0$, $y \neq -(2-\frac{2}{\gamma})$}, \label{eq:Lpnormsest}
\end{align}
and since $h$ has exactly one simple pole at $-\ii(2-\frac{2}{\gamma})$, there exists $\beta>0$ such that\footnote{In fact, $\beta = \frac{1}{4\pi (2-\frac{2}{\gamma})} \mathrm{e}^{-\theta_{\gamma}(-\ii(2-\frac{2}{\gamma}))}$, which is positive because $\gamma>2$ and $\theta_{\gamma}(-\ii(2-\frac{2}{\gamma})) =\frac{1}{2\pi} \eta_{g_{\gamma}}(\mathrm{e}^{(2-\frac{2}{\gamma})(\cdot)}) \in \RR$.}
\begin{align}\label{eq:betah-bound-1}
	|\beta h(z) - p_\gamma(z+2\ii)| \lesssim 1 \quad \mbox{for $z$ in a neighborhood of } -\ii(2-\tfrac{2}{\gamma}).
\end{align}
From~\eqref{eq:Lpnormsest} and \eqref{eq:betah-bound-1}, we see that the function $m (z) = \beta h(z-2\ii)$ is meromorphic in $S_2$ with $m-p_\gamma \in \mathbb{H}^{1,2}(S)$. Moreover, $m$ extends continuously to the boundary $\partial S_2$, and therefore satisfies assumption~\ref{it:mbound} from Lemma~\ref{lem:holomorphicEulerLagrange}. Consequently, it remains to show that $h$ solves
\begin{align}
	\signum h_{-2} = -h_0. \label{eq:elh}
\end{align}
To this end, note that by Remark~\ref{rem:blaschke},
 \begin{align*}
	&\sgn B_{\gamma}(x-2\ii) = \exp\left(-\ii \arctan\left(\frac{4-\frac{2}{\gamma}}{x}\right) + \ii \arctan\left(\frac{\frac{2}{\gamma}}{x}\right)\right),
 \intertext{and}
     &B_{\gamma}(x) = \exp\left(-2\ii \arctan\left(\frac{2-\frac{2}{\gamma}}{x}\right)\right),
\end{align*}
hence 
\begin{align*}
	\frac{\signum B_\gamma(x-2i)}{B_\gamma(x)} = \mathrm{e}^{\ii f_\gamma(x)},
	\quad \text{where} \quad
	f_\gamma(x) = 2 \arctan\biggr(\frac{2-\frac{2}{\gamma}}{x}\biggr) + \arctan\biggr(\frac{\frac{2}{\gamma}}{x}\biggr) - \arctan\biggr(\frac{4-\frac{2}{\gamma}}{x}\biggr).
\end{align*}
It follows that 
\begin{align*}
	\frac{\signum h_{-2}(x)}{h_0(x)} = \mathrm{e}^{\ii f_{\gamma}(x)} \,\mathrm{e}^{\ii \im \theta_{\gamma}(x+2\ii) - \theta_{\gamma}(x)}.
\end{align*}
Consequently, $h$ satisfies \eqref{eq:elh} if and only if
\begin{align}
	\ii \mathrm{Im}\, \theta_{\gamma}(x-2\ii) - \theta_{\gamma}(x) = -\ii f_\gamma(x) - \ii \pi a(x), \label{eq:thetaequality}
\end{align}
for some $a : \R \rightarrow 2 \Z +1$. 
Since the distributional Fourier transform of $f(x) = \arctan\frac{1}{x}$ is given by 
\begin{align*}
	\widehat{f}(k) =  -\pi \ii \frac{1-\mathrm{e}^{-|k|}}{k} \quad \text{for } k\in\RR,
\end{align*}
we have that 
\begin{align*}
	\widehat{f_{\gamma}}(k) = \ii \frac{g_{\gamma}(k)-2\pi}{k}.
\end{align*}
In particular, taking $a(x) = \signum(-x)$ with distributional Fourier transform $\widehat{a}(k) = 2 \ii \,\mathrm{p.v.} \frac{1}{k}$, 
we have 
\begin{align*}
	\frac{1}{2\pi} \widetilde{\eta}_{g_{\gamma}}(\mathrm{e}^{\ii x (\cdot)}) 
	&= \frac{1}{2\pi} \lim_{\epsilon\downarrow0}\int_{|k|\geq \epsilon} \frac{g_{\gamma}(k)}{k} \,\mathrm{e}^{\ii k x}\,\mathrm{d}k
	= \frac{\ii}{2\pi} \lim_{\epsilon\downarrow0}\int_{|k|\geq \epsilon} \frac{g_{\gamma}(k)}{k} \,\sin(kx)\,\mathrm{d}k  \\
	&=\frac{\ii}{2\pi} \lim_{\epsilon\downarrow0}\int_{|k|\geq \epsilon} \frac{g_{\gamma}(k)-2\pi}{k} \,\sin(kx)\,\mathrm{d}k + \frac{\ii}{2\pi} \lim_{\epsilon\downarrow0}\int_{|k|\geq \epsilon} \frac{2\pi}{k} \,\sin(kx)\,\mathrm{d}k \\
	&= \frac{1}{2\pi} \int_{\RR} \widehat{f_{\gamma}}(k)\,\sin(kx)\,\mathrm{d}k + \frac{1}{2\pi} \widehat{\pi a}(\sin(k (\cdot))) 
	= -\ii f_{\gamma}(x) - \ii \pi a(x).
\end{align*}
It then follows with \eqref{eq:phasefunctionid} that \eqref{eq:thetaequality} holds.
\end{proof}

\section{Further consequences for the CLR and LT bound} \label{sec:asymptotics}

\subsection{Asymptotics of LT and CLR upper bound}
In this section, we prove the asymptotic upper bounds for the LT and CLR constant presented respectively in Corollary~\ref{cor:improvedLTbound} and Table~\ref{tab:CLRbound}, which are based on the following lemma.
\begin{lemma}[Asymptotics for three-line problem] Let $\gamma >2$ and $h$ be the optimizer from Theorem~\ref{thm:optimizer}. Then we have
\begin{align}
    &\lim_{\gamma\downarrow 2} \frac{4\pi (\gamma-2)^{\gamma/2}}{\gamma^{\gamma/2}} \frac{\norm{h_{-2/\gamma}}_{L^\infty(\R)}^\gamma}{\norm{h_0}_{L^\infty(\R)}^{\gamma-2} \norm{h_{-1}}_{L^2(\R)}^2} = 1\label{eq:lowdlimit} \\
    &\lim_{\gamma \ra \infty}  4\pi \frac{\norm{h_{-2/\gamma}}_{L^\infty(\R)}^\gamma}{\norm{h_0}_{L^\infty(\R)}^{\gamma-2} \norm{h_{-1}}_{L^2(\R)}^2}  = 4\pi \mathrm{e}^2\biggr(\int_\R \frac{x^2+9}{x^2+1} e^{2\re \theta_\infty(x-\ii)} \mathrm{d}x\biggr)^{-1} \approx 5.342823,\label{eq:highdlimit}
\end{align}
where
\begin{align*}
    \re \theta_\infty(x-\ii) =  \int_0^\infty \frac{2 e^{-2k} - e^{-4k} + 1}{k (\cosh(2k) - 1)} \bigr(\cos(xk)\sinh(k) - k) \mathrm{d} k.
\end{align*}
\end{lemma}
\begin{proof}
    We start with a few preliminary observations. Note that for $x+\ii y \in \CC\setminus\{-\ii(2-\frac{2}{\gamma})\}$ we have 
    \begin{align}\label{eq:mod-blaschke}
        |B_{\gamma}(x+\ii y)|^2 = \frac{x^2 + \left(2-\frac{2}{\gamma}-y\right)^2}{x^2 + \left(2-\frac{2}{\gamma}+y\right)^2},
    \end{align}
    and $|\mathrm{e}^{\theta_{\gamma}(x+\ii y)}| = \mathrm{e}^{\Re \theta_{\gamma}(x+\ii y)}$ with 
    \begin{align}
        \re \theta_\gamma(x+\ii y) = - \frac{1}{\pi} \int_0^\infty \frac{g_{\gamma}(k)}{k (\cosh(2k) -1)} \bigr(\cos(xk) \sinh(yk)-yk\bigr)\mathrm{d}k. \label{eq:thetagamma}
    \end{align}
    It follows that $\|h_0\|_{L^{\infty}} = 1$, and 
    \begin{align*}
        \|h_{-1}\|_{L^2(\RR)}^2 
        &= \int_{\RR} \frac{x^2 + (3-\frac{2}{\gamma})^2}{x^2 + (1-\frac{2}{\gamma})^2} \mathrm{e}^{2\re \theta_{\gamma}(x-\ii)}\,\mathrm{d}x \\
        &= \frac{1}{1-\frac{2}{\gamma}} \int_{\RR} \frac{(1-\frac{2}{\gamma})^2 x^2 + (3-\frac{2}{\gamma})^2}{x^2 + 1} \mathrm{e}^{2\re \theta_{\gamma}((1-\frac{2}{\gamma})x - \ii)}\,\mathrm{d}x.
    \end{align*}
    In particular, by dominated convergence (recall \eqref{eq:expdecay}), we obtain 
    \begin{align}\label{eq:h-1-gamma-2}
        \left(1-\frac{2}{\gamma}\right) \|h_{-1}\|_{L^2(\RR)}^2 
        \stackrel{\gamma\downarrow 2}{\longrightarrow} 4 \mathrm{e}^{2 \re \theta_2(-\ii)} \int_{\RR} \frac{\mathrm{d}x}{x^2 + 1} 
        = 4\pi \mathrm{e}^{2 \re \theta_2(-\ii)}.
    \end{align}
    Moreover,
    the re-scaled boundary value $g(x) = \tfrac{2}{\gamma} h_0(\tfrac{2}{\gamma} x)$ is an optimizer of \eqref{eq:scaleinvariantdual} by construction (see the proof of Lemma~\ref{lem:holomorphicEulerLagrange}), in particular, $\norm{h_{-\frac{2}{\gamma}}}_{L^\infty} = \tfrac{\gamma}{2} \re \inner{\exp, \widehat{g}} = h(-\tfrac{2}{\gamma} \ii)$. Hence,
    \begin{align}
        \|h_{-\frac{2}{\gamma}}\|_{L^{\infty}} 
        &= |h(-\ii\tfrac{2}{\gamma})| 
        = \frac{1}{1-\frac{2}{\gamma}} \mathrm{e}^{\Re \theta_{\gamma}(-\ii\frac{2}{\gamma})},
    \end{align}
    so that 
    \begin{align}\label{eq:h-gamma-2}
         \left(1-\frac{2}{\gamma}\right)^{\gamma} \|h_{-\frac{2}{\gamma}}\|_{L^{\infty}}^{\gamma} 
         = \mathrm{e}^{\gamma \Re \theta_{\gamma}(-\ii\frac{2}{\gamma})}
         \stackrel{\gamma\downarrow 2}{\longrightarrow} \mathrm{e}^{2 \re \theta_2(-\ii)}.
    \end{align}
    Combining \eqref{eq:h-1-gamma-2} and \eqref{eq:h-gamma-2}, we have
    \begin{align*}
        \lim_{\gamma\downarrow 2} \frac{4\pi (\gamma-2)^{\gamma/2}}{\gamma^{\gamma/2}} \frac{\norm{h_{-2/\gamma}}_{L^\infty(\R)}^\gamma}{\norm{h_0}_{L^\infty(\R)}^{\gamma-2} \norm{h_{-1}}_{L^2(\R)}^2} 
        = \lim_{\gamma\downarrow 2} 4\pi \left(1-\frac{2}{\gamma}\right)^{\frac{\gamma}{2}-1} \frac{(1-\frac{2}{\gamma})^{\gamma}\|h_{-\frac{2}{\gamma}}\|_{L^{\infty}}^{\gamma}}{(1-\frac{2}{\gamma})\|h_{-1}\|_{L^2}^2}
        = 1,
    \end{align*}
    since $\lim_{\gamma\downarrow 2} \left(1-\frac{2}{\gamma}\right)^{\frac{\gamma}{2}-1} = 1$. 

    For the limit $\gamma \to \infty$, first observe that 
    \begin{align*}
        \re \theta_{\gamma}(-\ii \tfrac{2}{\gamma}) = \frac{1}{\pi} \int_{0}^{\infty} \frac{g_{\gamma}(k)}{k (\cosh(2k) -1)} \bigr(\sinh(\tfrac{2}{\gamma}k)-\tfrac{2}{\gamma}k\bigr)\mathrm{d}k,
    \end{align*}
    and since $\sinh(\tfrac{2}{\gamma}k) -\tfrac{2}{\gamma}k \lesssim (\tfrac{2}{\gamma}k)^3$ for $k \leq \tfrac{\gamma}{2}$ and $\sinh(\tfrac{2}{\gamma}k) -\tfrac{2}{\gamma}k \lesssim e^{\frac{2}{\gamma}k}$ for $k\geq \tfrac{\gamma}{2}$, we may bound 
    \begin{align*}
        \int_0^{\frac{\gamma}{2}} \frac{g_{\gamma}(k)}{k (\cosh(2k) -1)} \bigr(\sinh(\tfrac{2}{\gamma}k)-\tfrac{2}{\gamma}k\bigr)\mathrm{d}k 
        \lesssim \left(\frac{2}{\gamma}\right)^3 \int_0^{\infty} g_{\gamma}(k)\,\mathrm{d}k 
        = \left(\frac{2}{\gamma}\right)^3 \frac{\pi \gamma^3}{2\gamma^2-3\gamma+1},
    \end{align*}
    and, since $k\mapsto \mathrm{e}^{\frac{2}{\gamma}k}{k(\cosh(2k)-1)}$ is decreasing on $[\frac{\gamma}{2}, \infty)$, we have
    \begin{align*}
        \int_0^{\frac{\gamma}{2}} \frac{g_{\gamma}(k)}{k (\cosh(2k) -1)} \bigr(\sinh(\tfrac{2}{\gamma}k)-\tfrac{2}{\gamma}k\bigr)\mathrm{d}k 
        &\lesssim \frac{2}{\gamma} \frac{1}{\cosh(\gamma)-1} \int_{0}^{\infty} g_{\gamma}(k)\,\mathrm{d}k \\
        &= \frac{2}{\gamma} \frac{1}{\cosh(\gamma)-1} \frac{\pi \gamma^3}{2\gamma^2-3\gamma+1}.
    \end{align*}
    In particular, $\gamma \Re \theta_{\gamma}(-\ii\frac{2}{\gamma}) \to 0$ as $\gamma \to \infty$. Hence, using that $\left(1-\frac{2}{\gamma}\right)^{\gamma} \to \mathrm{e}^2$ as $\gamma \to \infty$, it follows that
    \begin{align}
        \lim_{\gamma\to \infty} \|h_{-\frac{2}{\gamma}}\|_{L^{\infty}(\RR)}^{\gamma} 
        = \lim_{\gamma\to \infty} \left(1-\frac{2}{\gamma}\right)^{-\gamma} \mathrm{e}^{\gamma \re \theta_{\gamma}(-\ii\frac{2}{\gamma})} 
        = \mathrm{e}^2.  \label{eq:limit3}
    \end{align}
    On the other hand, by dominated convergence we have
    \begin{align*}
        \re \theta_{\gamma}(x-\ii) 
        &= \frac{1}{\pi} \int_{0}^{\infty} \frac{g_{\gamma}(k)}{k (\cosh(2k)-1)} (\cos(x k) \sinh(k) - k)\,\mathrm{d}k  \\
        &\stackrel{\gamma\to\infty}{\longrightarrow} \frac{1}{\pi}\int_{0}^{\infty} \frac{g_{\infty}(k)}{k (\cosh(2k)-1)} (\cos(x k) \sinh(k) - k)\,\mathrm{d}k
        = \re \theta_{\infty}(x-\ii),
    \end{align*}
    with $g_{\infty}(k) = \pi \left( 2 \mathrm{e}^{-2 |k|} + 1 - \mathrm{e}^{-4|k|}\right)$. Therefore,
    \begin{align*}
        \norm{h_{-1}}_{L^2}^2 
        = \int_{\RR} \frac{x^2 + (3-\frac{2}{\gamma})^2}{x^2 + (1-\frac{2}{\gamma})^2} \mathrm{e}^{2\re \theta_{\gamma}(x-\ii)}\,\mathrm{d}x
        \stackrel{\gamma\to\infty}{\longrightarrow} \int_\R \frac{x^2+9}{x^2+1} \mathrm{e}^{2 \re \theta_\infty(x-\ii)} \mathrm{d} x.
    \end{align*}
    Combining this with \eqref{eq:limit3} then yields \eqref{eq:highdlimit}.
    \end{proof}
    
    The proof of the asymptotic bounds displayed in Table~\ref{tab:CLRbound} and equation~\eqref{eq:marginalimprovedLT} are now immediate consequences of Theorem~\ref{thm:threelines}, Theorem~\ref{thm:LTbound}, and equation~\eqref{eq:highdlimit} in the case of the CLR bound, and Theorem~\ref{thm:threelines}, Corollary~\ref{cor:improvedLTbound}, and equations~\eqref{eq:lowdlimit} and \eqref{eq:highdlimit} in the case of the LT bound.

\addtocontents{toc}{\protect\setcounter{tocdepth}{-1}}
\section*{Acknowledgements}
The authors thank Gero Friesecke and Dirk Hundertmark for helpful discussions. 

This work was initiated when TR was a visiting research fellow at TU München, funded by 
\emph{Deutsche Forschungsgemeinschaft (DFG – German Research Foundation) –- Project-ID 195170736 –- TRR109}.
TR would like to thank the Department of Mathematics at TUM for the kind hospitality.

TC acknowledges funding by the 
\emph{Deutsche Forschungsgemeinschaft (DFG – German Research Foundation) –- Project number 442047500} through the Collaborative Research Center "Sparsity and Singular Structures" (SFB 1481).
\addtocontents{toc}{\protect\setcounter{tocdepth}{2}}
\appendix
\section{Connection with maximal Fourier multipliers} \label{sec:app-Fourier}
In this section we show how the variational characterization of the Fourier transform of integrable functions from Lemma~\ref{lem:fourier-L1} can be used to reformulate the maximal Fourier multiplier bound derived in \cite[Theorem 2.1, Theorem 4.2]{HKRV23}.

To this end, let us first recall the setup of \cite{HKRV23}. Let $f, g : \R^d \rightarrow \R_+$ be Lebesgue measurable functions on $\R^d$, $d\geq 1$, $m : \R_+ \rightarrow \R_+$ be a continuous bounded function, and $B_{f,g,m}$ the operator whose integral kernel is given by
\begin{align}
    B_{f,g,m}(x,\eta) = \frac{1}{(2\pi)^{d/2}} \mathrm{e}^{\ii x \cdot \eta} m(f(x) g(\eta)). 
\end{align}
We then define the maximal operator associated to $B_{f,g,m}$ as\footnote{For the precise definition of the maximal operator, in particular measurability issues, we refer to \cite{HKRV23}.}
\begin{align}
    \mathcal{B}_{g,m}(\phi)(x) \coloneqq \sup_{f \geq 0} |B_{f,g,m}(\phi)(x)|, \quad x \in \R^d, \quad \mbox{for $\phi \in \mathcal{S}(\R^d)$.} \label{eq:maximalop}
\end{align}
One of the key results in \cite{HKRV23} concerns the boundedness of the maximal operator $\mathcal{B}_{g,m}$ on $L^2(\RR^d)$. 
\begin{theorem}[Theorem 2.1 in \cite{HKRV23}] \label{thm:maxopbound} Assume that
\begin{align*}
    m(t) = \int_{\R_+} m_1\left(\frac{t}{s}\right) m_2(s) \frac{\mathrm{d}s}{s}
\end{align*}
for some $m_1, m _2 \in L^2(\R_+, \tfrac{\mathrm{d} s}{s})$. Then the maximal operator $\mathcal{B}_{g,m}$ defined in \eqref{eq:maximalop} extends to a bounded (sub-linear) operator on $L^2(\R^d)$ with operator norm bounded by
\begin{align}
     \norm{ \mathcal{B}_{g,m}}_{L^2(\R^d)\to L^2(\R^d)} \leq \norm{m_1}_{L^2(\R_+,\frac{\mathrm{d}s}{s})} \norm{m_2}_{L^2(\R_+, \frac{\mathrm{d} s}{s})}.\label{eq:maximalopbound}
\end{align}
\end{theorem}

The proof of Theorem~\ref{thm:maxopbound} presented in \cite{HKRV23} relies on a clever combination of Plancherel's and Fubini-Tonelli's theorem, and the scaling invariance of the Haar measure $\frac{\mathrm{d}s}{s}$. The bound \eqref{eq:maximalopbound}, however, is not optimal since the convolutional decomposition of $m$ is not unique, and the product of the $L^2$ norms of $m_1$ and $m_2$ -- and hence the bound in \eqref{eq:maximalopbound} -- depends on this choice. In particular, this leaves some room for improvement by optimizing over $m_1$ and $m_2$.

It turns out that minimizing over the functions $m_1$ and $m_2$ yields precisely the $L^1$ norm of the Fourier transform of $m$, as proved in Lemma~\ref{lem:fourier-L1}. From this observation, we obtain the following improved version of Theorem~\ref{thm:maxopbound}.
\begin{theorem}[Improved maximal Fourier multiplier Bound]\label{thm:maxopboundimproved} Let $g: \R^d \rightarrow \R_+$ be a measurable function in $\R^d$, $d\geq 1$, and suppose that $m : \R_+ \rightarrow \R_+$ is a non-negative measurable function whose Fourier transform, defined as
\begin{align*}
    \widehat{m}(\omega) \coloneqq \int_0^{\infty} m(t) \mathrm{e}^{- \ii \log \omega \log t} \frac{\mathrm{d} t}{t},
\end{align*}
satisfies 
\begin{align}
    \norm{\widehat{m}}_{L^1(\R_+, \frac{\mathrm{d}\omega}{\omega})} = \int_0^\infty |\widehat{m}(\omega)| \frac{\mathrm{d} \omega}{\omega} < \infty. \label{eq:FTmultiplicative}
\end{align}
Then the maximal operator $\mathcal{B}_{g,m}$ defined in \eqref{eq:maximalop} extends to a bounded operator on $L^2(\R^d)$ with
\begin{align}
    \norm{\mathcal{B}_{g,m}}_{L^2(\R^d) \to L^2(\R^d)} \leq \frac{1}{2\pi} \norm{\widehat{m}}_{L^1(\R_+, \frac{\mathrm{d} s}{s})}.\label{eq:improvedmaximalopbound}
\end{align}
\end{theorem}
\begin{proof}
    First note that the exponential map $\exp : \R \rightarrow \R_+$ is a group isomorphism from the additive group to the multiplicative group, and pushes the Lebesgue measure forward to the measure $\frac{\mathrm{d}s}{s}$. Thus for any $m \in L^p(\R_+, \tfrac{\mathrm{d}s}{s})$, $1 \leq p \leq \infty$, there holds $\norm{\widetilde{m}}_{L^p(\R)} = \norm{m}_{L^p(\R_+, \frac{\mathrm{d}s}{s})}$, where $\widetilde{m}(x) \coloneqq m(\exp(x))$, and 
    \begin{align*}
        \int_{\R_+} m_1\left(\frac{t}{s}\right) m_2(s) \frac{\mathrm{d}s}{s} = \int_{\R} \widetilde{m_1}(\log t - x) \widetilde{m}_2(x) \mathrm{d} x.
    \end{align*}
    Moreover, we have $\widetilde{\widehat{m}}= \widehat{\widetilde{m}}$, where $\widehat{m}$ is defined via \eqref{eq:FTmultiplicative} and $\widehat{\widetilde{m}}$ is the standard Fourier transform of $\widetilde{m}$ on $\R$. In particular, the variational characterization of $\norm{\widehat{m}}_{L^1}$ in Lemma~\ref{lem:fourier-L1} can be transfered to the multiplicative group, i.e., 
    \begin{align*}
        \inf \biggr\{ \norm{m_1}_{L^2(\R_+,\frac{\mathrm{d} s}{s})} \norm{m_2}_{L^2(\R_+,\frac{\mathrm{d}s}{s})} : m_1*m_2 = m \biggr\} 
        = \|\widecheck{m}\|_{L^1(\R_+,\frac{\mathrm{d} s}{s})}
        = \frac{1}{2\pi} \norm{\widehat{m}}_{L^1(\R_+,\frac{\mathrm{d} s}{s})}.
    \end{align*}
To complete the proof, we can now apply Theorem~\ref{thm:maxopbound} to $\mathcal{B}_{g,m}$ (note that $m$ is continuous as its Fourier transform lies in $L^1(\R)$), optimize over $m_1$ and $m_2$, and use the above characterization.
\end{proof}

\section{Some technical details on the spaces \texorpdfstring{$\mathbb{H}^{p,q}$}{Hp,q}} \label{sec:app-hardy}
In this section, we give the proof of Lemma~\ref{lem:existenceboundary} on the existence of boundary values in the sense of distributions for functions in $\mathbb{H}^{p,2}(S)$.
\begin{lemma}[Uniform estimate]\label{lem:uniformest} Let $1\leq p \leq \infty$ and $h \in \mathbb{H}^{p,2}(S)$, then we have
\begin{align*}
|h(z)| \lesssim  \norm{h}_{\mathbb{H}^{p,2}} \bigr(y^{-\frac{1}{p}} + (1-y)^{-\frac12}\bigr) \quad \text{for any} \quad  z=x+\ii y \in S.
\end{align*}
\end{lemma}

\begin{proof}
First note that by the definition of $\mathbb{H}^{p,2}(S)$ in \eqref{eq:hardynorm}, we can find two functions $f, g: S \rightarrow \C$ such that $f+g = h$ and
\begin{align}
	\norm{f_y}_{L^p(\R)} \lesssim \norm{h}_{\mathbb{H}^{p,2}} (1-y)
	\quad\mbox{and}\quad 
	\norm{g_y}_{L^2(\R)} \lesssim \norm{h}_{\mathbb{H}^{p,2}} y, \label{eq:fgdecomposition}
\end{align}
for $y\in (0,1)$.
Consequently, by the mean value property of holomorphic functions and H\"older's inequality,  for any $z = x+iy \in S$ and $\rho < \min\{y, 1-y\}$ we have
\begin{align}
	|h(z)|   
	&= \frac{1}{\pi \rho^2} \int_0^\rho \biggr| \int_0^{2\pi} h(z+r\mathrm{e}^{it}) \mathrm{d}t \biggr| r\mathrm{d} r \leq \frac{1}{\pi \rho^2} \int_{B_\rho(z)} |f(u)| \mathrm{d}u + \frac{1}{\pi \rho^2} \int_{B_\rho(z)} |g(u)| \mathrm{d}u \nonumber \\
	&\leq \frac{(\pi \rho^2)^{1-\frac{1}{p}}}{\pi \rho^2} \biggr(\int_{y-\rho}^{y+\rho} \norm{f_w}_{L^p(\R)}^p \mathrm{d}w\biggr)^{\frac{1}{p}} + \frac{(\pi \rho^2)^{\frac12}}{\pi \rho^2} \biggr(\int_{y-\rho}^{y+\rho} \norm{g_w}_{L^2}^2 \mathrm{d} w \biggr)^{\frac12}\nonumber \\
	&\lesssim \norm{h}_{\mathbb{H}^{p,2}} \left[\frac{1}{\rho^{\frac{2}{p}}}\biggr(\int_{y-\rho}^{y+\rho} (1-w)^p \mathrm{d}w\biggr)^{\frac{1}{p}} + \frac{1}{\rho} \biggr(\int_{y-\rho}^{y+\rho} w^2\mathrm{d}w \biggr)^{\frac{1}{2}}\right]. \label{eq:hardyest}
\end{align}
Hence by taking the limit $\rho \ra \min\{y,(1-y)\}$, the first term in \eqref{eq:hardyest} can be bounded by $y^{-\frac{1}{p}}$ and the second term by $(1-y)^{-\frac12}$. 
\end{proof}

For the proof of Lemma~\ref{lem:existenceboundary}, we shall need the following regularity result.

\begin{lemma}[Absolute continuity of measures with exponentially weighted Fourier transform]\label{lem:abscont} Let $\mu \in \mathcal{M}(\R)$ and suppose that $\widehat{\mu} \in L^2_{\gamma}$ for some $\gamma \in \R$. Then $\mu$ is absolutely continuous with respect to the Lebesgue measure.
\end{lemma}

The proof of Lemma~\ref{lem:abscont} relies on a classical result of F.\ and M.\ Riesz \cite{RR16} on \emph{analytic} measures, that is, Borel measures $\mu$ on $\RR$ with the property that $\widehat{\mu}(k) = 0$ for all $k\leq 0$. The following version can be found in \cite[Lemma 13.4]{Mas09}. 

\begin{lemma} \label{lem:Riesz} Let $\mu \in \mathcal{M}(\R)$. Then the following are equivalent:
\begin{enumerate}[label=(\roman*)]
\item $\mu$ is analytic;
\item for any $z\in \C_+ = \{z\in\CC: \mathrm{Im}(z) >0\}$, 
    \begin{align*}
	   \int_\R \frac{\mu (\mathrm{d}x)}{x-\overline{z}}  = 0; 
    \end{align*}
\item $\mu(\mathrm{d} x) = u(x) \mathrm{d}x$, where $u \in L^1(\R)$ with $\widehat{u}(k) = 0$ for all $k\leq 0$.
\end{enumerate}
\end{lemma}
We refer to \cite[Sections 5.5 and 13.2]{Mas09} for more details and the proof.

\begin{proof}[Proof of Lemma~\ref{lem:abscont}] The case $\gamma=0$ follows from the fact that the Fourier transform is an isomorphism on $L^2(\R)$. Moreover, by considering $\tilde{\mu}(A) = \mu(-A)$ instead of $\mu$, we may assume that $\widehat{\mu} \in L^2_\gamma$ with $\gamma>0$. 

Define the function  
\begin{align*}
	\widehat{m}(k) \coloneqq \widehat{\mu}(k) 1_{\{k \leq 0\}},
	\end{align*}
then from the assumption $\widehat{m} \in L^2_\gamma$ and H\"older's inequality, we have that $\widehat{m} \in L^1(\R) \cap L^2(\R)$, which implies that $m \in C_0(\R)\cap L^2(\R)$. In particular, the measure
\begin{align*}
	\mu_1 \coloneqq \mu - m \mathrm{d}x
\end{align*}
has a well-defined Fourier transform that vanishes identically for $k\leq 0$. 

Now note that for any $z \in \C_+$, the function $x\mapsto r_z(x)\coloneq \frac{1}{2\pi \ii} \frac{1}{x-z}$ belongs to $C_0(\R) \cap L^2(\R)$, and its Fourier transform $\widehat{r_z}(k) = \mathrm{e}^{-\ii z k} 1_{\{ k \leq 0\}}$ lies in $L^1(\R) \cap L^2(\R)$. 
Let $\varphi\in\mathcal{S}(\RR)$ and define $\varphi_{\epsilon}\coloneq \epsilon^{-1} \varphi(\frac{\cdot}{\epsilon})$. Then $\mu_1*\varphi_{\epsilon} = \mu \ast \varphi_{\epsilon} + m \ast \varphi_{\epsilon} \in L^1(\R) + L^2(\R)$ and we have 
\begin{align*} 
    \mu*\varphi_{\epsilon} \stackrel{*}{\rightharpoonup} \mu \quad \mbox{in $\mathcal{M}(\R)$}\quad {and}\quad m*\varphi_{\epsilon} \to m\quad \mbox{in $L^2(\R)$ as $\epsilon\downarrow 0$.}
\end{align*}
In particular 
\begin{align*}
	\frac{1}{2\pi \ii} \int_\R \frac{\mu_1(\mathrm{d}x)}{x-\overline{z}}  
    &= \int_{\RR} \overline{r_z(x)}\,\mu_{1}(\mathrm{d}x) 
    = \lim_{\epsilon \downarrow 0} \int_{\RR} \overline{r_z(x)}\,\mu_{1}*\varphi_{\epsilon}(x)\,\mathrm{d}x.
\end{align*} 
Parseval's formula then implies that for any $\epsilon>0$
\begin{align*}
    \int_{\RR} \overline{r_z(x)}\,\mu_{1}*\varphi_{\epsilon}(x)\,\mathrm{d}x
    = \frac{1}{2\pi} \int_{\RR} \overline{\widehat{r_z}(k)}\,\widehat{\mu_1}(k) \widehat{\varphi_{\epsilon}}(k)\,\mathrm{d}k
    = 0,
\end{align*}
since $\widehat{r_z}(k)$ and $\widehat{\mu_1}(k)$ have disjoint support. It follows that for all $z\in \CC_+$,
\begin{align*}
	\frac{1}{2\pi \ii} \int_\R \frac{\mu_1(\mathrm{d}x)}{x-\overline{z}}  
    &= 0.
\end{align*}
From the identity $(\overline{z_1}-\overline{z_2})(x-\overline{z_1})^{-1}(x-\overline{z_2})^{-1} = (x-\overline{z_1})^{-1} - (x-\overline{z_2})^{-1}$ for any $x\in\RR$, $z_1, z_2 \in \CC_+$, we conclude that
\begin{align*}
	\int_\R \frac{\mu_1(\mathrm{d}x)}{(x-\overline{z_1})(x-\overline{z_2})} = 0 \quad \mbox{for any $z_1,z_2 \in \C_+$.}
\end{align*}
In particular, the measure $\mu_2(\mathrm{d} x) \coloneqq (x+\ii)^{-1}  \mu_1 (\mathrm{d}x)$ is a bounded Radon measure satisfying the assumptions of Lemma~\ref{lem:Riesz}. Therefore, $\mu_2$ is absolutely continuous with respect to Lebesgue measure on $\RR$, which implies that $\mu(\mathrm{d}x) = (x+\ii) \mu_2(\mathrm{d}x)+ m \mathrm{d}x$ is also absolutely continuous.
\end{proof}

\begin{proof}[Proof of Lemma~\ref{lem:existenceboundary}]
The first step is to show that for any $\phi \in C_c^\infty(\R)$, the function
\begin{align}
 y\mapsto \inner{h_y, \widehat{\phi}_{y_0-y}} \quad \mbox{ is independent of $y \in (0,1)$.} \label{eq:closedcurveintegral}
\end{align}
For this, we notice that the function $z\mapsto \overline{h(\overline{z})} \widehat{\phi}(z + \ii y_0)$ is holomorphic on $S_-$, as $\widehat{\phi}$ is entire by the Paley-Wiener theorem. Hence Cauchy's Integral Theorem implies that
\begin{multline*}
	\int_{-R}^R \biggr(\overline{h_y(x)} \widehat{\phi}_{y_0-y}(x)  -\overline{h_{y'}(x)}\widehat{\phi}_{y_0-y'}(x)\biggr)\mathrm{d} x  \\
	+ \int_{y'}^y \biggr(\overline{h_{w}(-R)}\widehat{\phi}_{y_0-w}(-R) - \overline{h_{w}(R)}\widehat{\phi}_{y_0-w}(R)\biggr) \mathrm{d} w = 0,
\end{multline*}
for any $y, y' \in (0,1)$.
In the limit $R\ra \infty$, the first term converges to $\inner{h_y, \widehat{\phi}_{y_0-y}} - \inner{h_{y'},\widehat{\phi}_{y_0-y'}}$, while by the uniform (with respect to $R$) control on $h_w(\pm R)$ from Lemma~\ref{lem:uniformest} and the fast decay of $\phi_w(\pm R)$, the second term vanishes, which proves \eqref{eq:closedcurveintegral}.

Let us treat the case $p\in (1, \infty]$ first. By definition of the space $\mathbb{H}^{p,2}(S)$, we can decompose $h=f+g$ such that \eqref{eq:fgdecomposition} holds, i.e.\ $\norm{f_y}_{L^p(\R)} \lesssim \norm{h}_{\mathbb{H}^{p,2}} (1-y)$ and $\norm{g_y}_{L^2(\R)} \lesssim \norm{h}_{\mathbb{H}^{p,2}} y$. In particular, as $y\downarrow 0$, we have $\|g_y\|_{L^2}(\RR) \to 0$ and $\|f_y\|_{L^p(\RR)}\lesssim 1$. Hence, $g_y \to 0$ strongly in $L^2(\RR)$ as $y\downarrow 0$, and there exists a subsequence $y_n \downarrow 0$ and a function $h_0\in L^p(\RR)$ such that $f_{y_n} \rightharpoonup h_0$ weakly in $L^p$ (respectively weakly-* in $L^{\infty}$ for $p=\infty$) as $y_n \downarrow 0$. We claim that the limit $h_0$ is independent of the subsequence $(y_n)$. 
Indeed, there holds 
\begin{align}
	\inner{h_y - h_0, \widehat{\phi}_0} &= \inner{h_y,\widehat{\phi}_{-y}} - \inner{h_0, \widehat{\phi}_0}  - \inner{h_y, \widehat{\phi}_{-y}-\widehat{\phi}_0} \nonumber \\
	&\stackrel{\eqref{eq:closedcurveintegral}}{=}  \inner{h_{y_n}, \widehat{\phi}_{-y_n}} - \inner{h_0,\widehat{\phi}_0} - \inner{h_y, \widehat{\phi}_{-y}-\widehat{\phi}_0} \nonumber\\
	&= \inner{h_{y_n}-h_0, \widehat{\phi}_0} + \inner{h_{y_n}, \widehat{\phi}_{-y_n}-\widehat{\phi}_0} - \inner{h_y,\widehat{\phi}_{-y}-\widehat{\phi}_0}  \label{eq:auxest}
\end{align}
for any $\phi \in C^\infty_c(\R)$. Note that for any $y\in(0,1)$ we may bound 
\begin{align*}
    |\inner{h_y,\widehat{\phi}_{-y}-\widehat{\phi}_0}| 
    &\leq |\inner{f_y,\widehat{\phi}_{-y}-\widehat{\phi}_0}|+|\inner{g_y,\widehat{\phi}_{-y}-\widehat{\phi}_0}| \\
    &\leq \|f_y\|_{L^p(\RR)} \|\widehat{\phi}_{-y}-\widehat{\phi}_0\|_{L^{\frac{p}{p-1}}(\RR)} + \|g_y\|_{L^2(\RR)} \| \widehat{\phi}_{-y}-\widehat{\phi}_0 \|_{L^2(\RR)} \to 0
\end{align*}
as $y\downarrow 0$, which follows from the uniform bounds $\sup_{0<y<1}\left(\|f_y\|_{L^p} + \|g_y\|_{L^2}\right) \lesssim 1$ and the convergence $\widehat{\phi}_{-y} \to \widehat{\phi}_0$ in $L^1(\RR) \cap L^{\infty}(\RR)$ as $y\downarrow 0$. 
Hence, letting $y_n \downarrow 0$ in \eqref{eq:auxest}, we obtain $\inner{h_y - h_0, \widehat{\phi}_0} = - \inner{h_y,\widehat{\phi}_{-y}-\widehat{\phi}_0} \to 0$ as $y\downarrow 0$.  
The argument for the limit $y\uparrow 1$ is analogous. Since $\mathcal{C}^{\infty}_c(\RR)$ is dense in $\mathcal{S}(\RR)$, we conclude that equation \eqref{eq:boundaryvalues} holds for any Schwartz function $\varphi \in\mathcal{S}(\RR)$.

For the case $p=1$, we can only infer the existence of a finite measure $h_0 \in \mathcal{M}(\RR)$ on $\RR$ such that $f_{y_n} \stackrel{*}{\rightharpoonup} h_0$ weakly-* in the sense of measures as $y_n\downarrow 0$. Therefore, it remains to show that $h_0\in L^1(\RR)$. To this end, notice that \eqref{eq:closedcurveintegral} extends to $y=0$ and $y=1$, i.e.,
\begin{align}
	\inner{h_0, \widehat{\phi}_0} = \inner{h_y, \widehat{\phi}_{1-y}} \quad \mbox{for any $\phi \in C_c^\infty(\R)$ and $0\leq y \leq 1$.} \label{eq:extendedlineintegral}
\end{align}
This follows from the fact that $\widehat{\phi}_{y} \ra \widehat{\phi}_{y_0}$ strongly in $L^1\cap L^\infty$ and $h_y \ra h_{y_0}$ weakly-$*$ in $\mathcal{M}(\R)$ as $y \ra y_0$ for any $0 \leq y_0 \leq 1$, and therefore proves \eqref{eq:cauchytheorem}.
We can now appeal to Lemma~\ref{lem:abscont} to conclude the case $p=1$. Indeed, since $\widehat{\phi}_{-1}$ is the Fourier transform of $\phi \exp_{-1}$, Parseval's formula implies that $\widecheck{h_1} = \widecheck{h_0} \exp_{-1} \in L^2(\R)$. In particular, $h_0$ has a Fourier transform in $L^2_2$, which implies that $h_0 \in L^1(\R)$ by Lemma~\ref{lem:abscont}.

\end{proof}

\section{Existence and uniqueness of primal optimizer}
\label{sec:app-primaloptimizer}
In this section, we prove that the primal problem~\eqref{eq:primalL1} admits a unique (up to symmetries) minimizer in $L^1(\R)$. In addition, we show that the dual problem can be restricted to the space of continuous functions vanishing at infinity. Precisely, we shall prove the following theorem.

\begin{theorem}[Primal optimizer]\label{thm:primalopt} Let $\gamma>2$ and $M_\gamma$ be defined in \eqref{eq:primalL1}. Then we have
\begin{align}
	M_\gamma = \min_{m \in L^1(\R)} \norm{m}_{L^1}^{\gamma-2} \norm{\widehat{m} - \exp}_{L^2_{\gamma}}^2 = \frac{4(\gamma-2)^{\gamma-2}}{(2\pi)^{\gamma-2} \gamma^\gamma}\sup_{\substack{g \in C_0(\R) \\ \widehat{g} \in L^2_{-\gamma}}} \biggr\{ \frac{|\inner{\exp,\widehat{g}}|^\gamma}{\norm{g}_{L^\infty}^{\gamma-2} \norm{\widehat{g}}_{L^2_{-\gamma}}^2}\biggr\},
\end{align}
i.e., the infimum over $L^1(\R)$ is attained, and the minimizer is unique up to the transformation $m(x) \mapsto m(x) \alpha^{-\ii x-1}$ for $\alpha >0$.
\end{theorem}

For the proof of Theorem~\ref{thm:primalopt}, we shall use the following lemma.
\begin{lemma} \label{lem:doubledual} Let $F_\gamma^\ast : C_0(\R) \rightarrow (-\infty, \infty]$ be the functional defined in Lemma~\ref{lem:dualFgamma}, but restricted to $C_0(\R)$. Then its Fenchel-conjugate $F_\gamma^{\ast\ast} : \mathcal{M}(\R) \rightarrow [0,\infty]$ is given by
\begin{align}
	F_\gamma^{\ast \ast}(\mu) = \begin{dcases} \norm{\widehat{\mu}-\exp}_{L^2_\gamma}^2, \quad &\mbox{if $\widehat{\mu}-\exp \in L^2_\gamma$,} \\
	+\infty, \quad \mbox{otherwise,} \end{dcases},
\end{align}
where $\mathcal{M}(\R)$ denotes the space of bounded Radon measures.
\end{lemma}

\begin{proof} From similar calculations as in the proof of Lemma~\ref{lem:dualFgamma}, 
 we find
\begin{align}
	F_\gamma^{\ast \ast}(\mu) &= \sup_{\substack{g \in C_0(\R) \\ \widehat{g} \in L^2_{-\gamma}} } \sup_{\alpha \in \C}\biggr\{ \re \alpha \biggr(\inner{\mu, g} - \frac{1}{2\pi} \inner{\exp, \widehat{g}}\biggr) -  \frac{|\alpha|^2}{16\pi^2} \norm{\widehat{g}}_{L^2_{-\gamma}}^2\biggr\} \nonumber \\
	&= \sup_{\substack{g \in C_0(\R) \\ \widehat{g} \in L^2_{-\gamma}}} \frac{|2 \pi \inner{\mu, g} - \inner{\exp,\widehat{g}}|^2}{\norm{\widehat{g}}_{L^2_{-\gamma}}^2} \label{eq:bi-dualF}
\end{align}
Now note that from Parseval's identity we have
\begin{align}
	2\pi \inner{\mu, g} - \inner{\exp, \widehat{g}} = \inner{ \widehat{\mu} - \exp , \widehat{g}}, \label{eq:parseval2}
\end{align}
for any $g \in C_0(\R)$ with $\widehat{g} \in C_c^\infty(\R)$. Thus by the Riesz representation theorem (recall that $C_c^\infty(\R)$ is dense in $L^2_\gamma(\R)$), the supremum in \eqref{eq:bi-dualF} is finite if and only if $\widehat{\mu} - \exp \in L^2_\gamma(\R)$. In this case we have
\begin{align*}
    \sup_{\substack{g \in C_0(\R) \\ \widehat{g} \in L^2_{-\gamma}}} \frac{|2 \pi \inner{\mu, g} - \inner{\exp,\widehat{g}}|^2}{\norm{\widehat{g}}_{L^2_{-\gamma}}^2} = \sup_{\substack{g \in C_0(\R) \\ \widehat{g} \in L^2_{-\gamma}}} \frac{|\inner{ \widehat{\mu} - \exp , \widehat{g}}|^2}{\norm{\widehat{g}}_{L^2_{-\gamma}}^2} = \norm{\widehat{\mu} - \exp}_{L^2_\gamma}^2,
\end{align*}
which completes the proof.
\end{proof}

\begin{proof}[Proof of Theorem~\ref{thm:primalopt}] From Lemma~\ref{lem:doubledual} and the Fenchel-Rockafellar duality theorem applied to the functional $F_\gamma^\ast + \norm{\cdot}_{L^\infty}$, we obtain
\begin{align}
	\widetilde{M}_\gamma \coloneqq \min_{\substack{\mu \in \mathcal{M}(\R) \\ \norm{\mu} \leq 1}} \norm{\widehat{\mu}-\exp}_{L^2_\gamma}^2  
    = \sup_{g \in C_0(\R)} \left( \norm{g}_{L^\infty} + \frac{1}{16\pi^2} \norm{\widehat{g}}_{L^2_{-\gamma}}^2 + \frac{1}{2\pi} \re \inner{\widehat{g},e^{(\cdot)}} \right) 
    \geq M_\gamma,
\end{align}
where the last inequality follows from Lemma~\ref{lem:dualproblem}, and a minimizer in $\mathcal{M}(\R)$ exists and is unique by the strict convexity of $\norm{\cdot}_{L^2_\gamma}^2$.

To conclude the proof, just note that any $\mu \in \mathcal{M}(\R)$ satisfying $\widehat{\mu} - \exp \in L^2_\gamma(\R)$ must be absolutely continuous with respect to the Lebesgue measure by Lemma~\ref{lem:abscont}, and therefore, the variational problem over $\mathcal{M}(\R)$ coincides with problem~\eqref{eq:primalL1}.
\end{proof}
\section*{Data availability}
No datasets were generated or analysed during the current study.

\section*{Competing interests}

The authors have no competing interests to declare that are relevant to the content of this article.

\bigskip
\end{document}